\RequirePackage{fix-cm}
\documentclass[smallcondensed,numbook]{svjour3}       
\smartqed   
\usepackage{graphicx}
 \usepackage{mathptmx}     
\usepackage{amsmath}        % Required by environment ``split''.
\usepackage{amssymb}        % Required by \mathbb.
\usepackage{multirow}       % Required by multirows in tables.
\usepackage[mathscr]{euscript}     % Use Euler mathscr font for functionals.
\usepackage[justification=centering]{caption}     % Make captions centering. 
\usepackage{type1cm,type1ec}    
\usepackage[  
pdfpagelabels,hypertexnames=true,
plainpages=false,
naturalnames=false]{hyperref}
\usepackage{color}
\definecolor{darkblue}{rgb}{0,0.1,0.5}
\hypersetup{colorlinks,
linkcolor=darkblue,
anchorcolor=darkblue,
citecolor=darkblue} 
\DeclareMathAlphabet{\mathcal}{OMS}{cmsy}{m}{n}    % Use standard calligrafic font of mathcal, instead of the rsfs one.
\DeclareMathAlphabet{\mathpzc}{OT1}{pzc}{m}{it}   
\renewcommand{\theenumi}{\textnormal{\alph{enumi}.)}}   % Change the display of the counter \theenumi to alphabets.
    
\makeatletter
\@ifundefined{proof}{ 
}{

}
\makeatother

\usepackage{booktabs} 
\usepackage{amsthm}

\usepackage{mathrsfs}
\usepackage{adjustbox}
\usepackage{tikz}
\usetikzlibrary{positioning}
\usetikzlibrary{chains}
\usepackage{tikz-cd}
\usepackage{pgfplots}
 
\newtheorem{assumption}{Assumption}

\usepackage{bm}

\def \x {{\bm{x}}}

\usepackage{amsthm}

\def\be{\begin{equation}}
\def\ee{\end{equation}}

\numberwithin{equation}{section} 

\usepackage{graphicx}     
\usepackage{subcaption}

\usepackage{algorithm}
\usepackage{algorithmic}



\begin{document}

\title{Privacy-Preserving Black-Box Optimization (PBBO): Theory and the Model-Based Algorithm DFOp}

\titlerunning{PBBO \& DFOp}        % if too long for running head

\author{Pengcheng Xie}


\institute{Applied Mathematics and Computational Research Division, 
Lawrence Berkeley National Laboratory, Berkeley, CA 94720, USA.\\
This work was partially supported by Laboratory Directed Research and
Development (LDRD) funding from Lawrence Berkeley National Laboratory and by the U.S.
Department of Energy, Office of Science, Office of Advanced Scientific Computing Research applied mathematics program (SEAZOTIE) under Contract Number DE-AC02-05CH11231. 
}

\date{ } 

\maketitle

\begin{abstract}
\begin{sloppypar}
    
This paper focuses on solving unconstrained privacy-preserving black-box optimization (PBBO),  its corresponding least Frobenius norm updating of quadratic models, and the differentially privacy mechanisms for PBBO.  Optimization problems with transformed/encrypted objective functions aim to minimize $F(\x)$, which is encrypted/transformed/encrypted to $F_k(\x)$ as the output at the $k$-th iteration. A new derivative-free solver named DFOp, with its implementation, is proposed in this paper, which has a new updating formula for the quadratic model functions. The convergence of DFOp for solving problems with transformed/encrypted objective functions is given. Other analyses, including the new model updating formula and the analysis of the transformation's impact to model functions are presented. We propose two differentially private noise-adding mechanisms for privacy-preserving black-box optimization. Numerical results show that DFOp performs better than compared algorithms. To the best of our knowledge, DFOp is the first derivative-free solver that can solve black-box optimization problems with step-encryption and privacy-preserving black-box problems exactly, which also tries to answer the open question about the combination of derivative-free optimization and privacy.  
\end{sloppypar}

\end{abstract}

\section{Motivation and Introduction}\label{Introduction}
 
 In many practical optimization scenarios, the objective function is not available in an explicit analytic form and may only be accessed through costly evaluations, such as those arising from complex simulations or physical experiments. In such settings, derivative information is often unavailable or unreliable, making classical gradient-based optimization techniques inapplicable. Problems of this type are commonly referred to as {derivative-free optimization} problems. 
Derivative-free optimization methods are specifically designed to operate without requiring gradient or higher-order derivative information, relying instead on function evaluations alone. These methods have been extensively studied and developed over the past decades; comprehensive overviews can be found in the monographs by Conn, Scheinberg, and Vicente~\cite{conn2009introduction}, as well as by Audet and Hare~\cite{audet2017derivative}. In this work, we focus on a class of derivative-free optimization problems in which the objective function undergoes transformation or encryption, posing additional challenges for efficient optimization.

\subsection{Transformed/encrypted objective functions and black-box optimization problems}

\begin{sloppypar}
In this paper, we are interested in the black-box optimization problems with transformed/encrypted objective functions (see more details of privacy-preserving optimization in \cite{lowy2023optimal,lowy2023private,lowy2023private2}). We will firstly present details of problems with transformed/encrypted objective functions, and then illustrate private problems before we come up with the new solver DFOp in Section \ref{DFOp}.
\end{sloppypar}

{%\color{red}
We consider unconstrained black-box optimization problems with transformed or encrypted objective functions of the form
\begin{equation}\label{fx}
\min_{\boldsymbol{x} \in \mathbb{R}^n} F(\boldsymbol{x}),
\end{equation}
where the objective function can only be accessed through function evaluations. At each iteration, the black box provides function values at a finite number of query points, and all queries within the same iteration share an identical transformation of the original objective function. Specifically, at iteration $k$, the original objective function $F$ is transformed into a modified function $F_k$ \cite{xie2024derivative}, and the optimization algorithm only observes evaluations of $F_k$.

\begin{definition}
Let $T_k:\mathbb{R}\rightarrow\mathbb{R}$ denote a transformation applied at iteration $k$. For a given function $F$ and any $\boldsymbol{x}\in\mathbb{R}^n$, we define
\[
F_k(\boldsymbol{x}) = T_k(F(\boldsymbol{x})).
\]
We refer to $F_k$ as the {transformed} or {encrypted} objective function corresponding to $F$ at iteration $k$.
\end{definition}

The transformation $T_k$ may vary from one iteration to another and is not assumed to satisfy any particular structural property. Such transformed or encrypted objective functions naturally arise in a variety of settings, including private or secure black-box optimization, where different transformations can be interpreted as different privacy-preserving mechanisms applied to the objective function. In particular, transformations induced by additive or multiplicative noise are closely related to concepts in differential privacy. 
More broadly, optimization problems involving transformed objective functions appear in several practical contexts. For instance, optimization problems with iteration-dependent regularization terms can be viewed as special cases of black-box optimization with transformed objectives. A representative example is the derivative-free trust-region framework for composite nonsmooth optimization proposed by Grapiglia, Yuan, and Yuan~\cite{Grapiglia2016}, which fits naturally into this setting.}

\subsection{Private black-box optimization and differential privacy}
\label{Private-BBO}

\begin{sloppypar} 
For some optimization problems in practice, personal information is supposed to be protected, while private information may leak out by adding a query. Thus the function values require the encryption. Besides, Liu et al. \cite{9186148} proposed the following open questions. 
\begin{itemize}
\item What roles does black-box/zeroth-order\footnote{Zeroth-order optimization is a subset of derivative-free optimization, and zeroth-order optimization methods solve black-box problems similarly to gradient-based methods.} optimization play when solving privacy-preserving problems?
\item How can we design derivative-free optimization algorithms with privacy guarantees?
\end{itemize}
\end{sloppypar}

We propose the definition of private black-box optimization based on  the characteristics of private problems and data privacy protection in practice. The general form of the unconstrained composite private black-box optimization problem\footnote{Constrained private black-box optimization problems will be considered in the future. Our results can extend in a straightforward way to constrained problems.} can be formulated as 
\begin{equation}\label{1.1}
\min_{x \in \mathbb{R}^n} \ F({x})=f(x)+h(x),
\end{equation}
where $f$ is a black-box function, and $h$ is a private black-box function. The evaluation of $f$ is expensive, but its function value is public and exact. At the $k$-th iteration, $h$ is encrypted to $h_k$ by adding noises, and the true function value of $h$ at each corresponding point is private. In view of the characteristics above, $f$ is called the public black-box function. The term public indicates the fact that the true function value is known by the public.  There may be some important information included in $h$, such as trade secrets, customer personal information and so on. When solving this kind of optimization problems, $f$ and $h$ should be called as few times as possible, for the expensive evaluation cost on $f$ in time, computation and other aspects, as well as the requirement of the encryption and the evaluation cost on $h$.

The encryption in $h$ may disturb the gradient significantly. Therefore, algorithms using gradient values may have numerical instability or even non-convergence. Fortunately, derivative-free methods do not use gradient values, which means that the effect of gradient error is avoid. In this sense, derivative-free methods have advantages over gradient-based ones for solving private optimization problems. The limitation of the function evaluations for protecting the true value of the function $h$ is similar to the need of an efficient derivative-free algorithm for solving expensive black-box optimization problems. Therefore, it is natural to modify derivative-free methods to solve private black-box optimization problems.

We note that grey-box optimization \cite{greybox} and private black-box optimization \cite{xieyuannew1,10.1093/imanum/drae106} are also related to each other. Grey-box optimization usually contains black-box constraints and the constraints with known forms or some other information of the objective function, which need to be encrypted in some practical cases. 
 
We next describe two representative examples that illustrate the structure and characteristics of private black-box optimization problems.

\begin{example}
The first one is the cloud-based distributed optimization problem, which aims to minimize the local and cloud-based composite objective functions, and to protect the privacy of the corresponding objective functions. The function $f$ denotes the exact and local information saved in the user’s personal computer. 
The function $h$ denotes the evaluation information obtained from the cloud, and others' privacy has been encrypted and saved in this part. One can see the work of Wang et al. \cite{inproceedingsDP} for more details.
\end{example}
\begin{example}
We can also find private black-box optimization problems in personal health fields, where $f$ is usually a public health evaluation, and $h$ denotes private information.  The individual health evaluation is contained in the function $h$. One can see medical information literature for more details, such as the work of Liu, Huang and Liu \cite{liu2015secure}.  
\end{example}
To protect the privacy from being attacked by adding a query, researchers in data science and cryptography have  proposed the concept of differential privacy and differentially private mechanisms\footnote{Details are in Section \ref{Solving private black-box optimization problems with DFOp}.}. As a kind of noise-adding encryption method, differentially private mechanisms can protect the queried data and reduce the risk of the record being recognized as much as possible. Therefore, they have been widely adopted in network security, and keep gaining momentum in implementation owing to the relatively less complex computation.

Differentially private mechanisms have been applied in different areas \cite{privacy13,privacy18,privacy7,privacy19,privacy5,privacy37,privacy35,privacy8,Kasiviswanathan2008}. The United States Census Bureau applies differentially private mechanisms to present the commuting mode \cite{census}. Companies, such as Google \cite{nnn-14}, Apple \cite{apple,nnn-1}, Microsoft \cite{nnn-9} and Alibaba \cite{nnn-32}, also apply such mechanisms as important tools to protect the users' privacy. Moreover, some work about the differentially private Bayesian optimization has been discussed by Kusner et al. \cite{diffpri}. 
 
\section{New Least Frobenius Norm Updating of Quadratic Models and DFOp}
\label{DFOp}
\begin{sloppypar} 
We start this section with model-based derivative-free optimization method (recent work can be seen from \cite{roberts2025introductioninterpolationbasedoptimization,xieyuannew1,XIE2025116146,Xie16122025,xie2025remuregionalminimalupdating,CartisFialaBenRob}) and the popular derivative-free optimization software NEWUOA \cite{powell2006newuoa}. NEWUOA is iterative, in which a quadratic model function $Q \approx F$ is required for adjusting the variables. In NEWUOA's framework and analysis, $Q_k$ denotes the $k$-th model function, and $X_k$ denotes the interpolation set at the $k$-th iteration. The number of interpolation points, denoted by $m$, satisfies that $n+2 \le m \le \frac{1}{2}(n+1)(n+2)$, which is flexible in the interpolation. Notice that the coefficients of the quadratic model function $Q$ are the symmetric Hessian matrix $\nabla^2Q$, the gradient vector $\nabla Q$ and the constant term, whose freedom is $\frac{1}{2}(n+1)(n+2)$ together, which is $\mathcal{O}(n^2)$. When the problem's dimension, we say $n$, is large, if we directly determine the coefficients of the model function $Q$ by solving interpolation equations $Q(x_i)=F(x_i), \ i=1,\ldots,\frac{1}{2}(n+1)(n+2)$, the number of function evaluations is too large. To reduce the numbers of calling functions, Powell \cite{powell2006newuoa} advised that we can use fewer interpolation points to obtain the quadratic model function. In order to uniquely determine the coefficients of $Q$, the iterative quadratic model $Q_k$ is designed to be the solution of
\end{sloppypar}
\begin{equation}\label{leastfrob}
\begin{aligned}
\underset{Q}{\operatorname{\min}}\ \left\Vert\nabla^{2} Q-\nabla^{2} Q_{k-1}\right\Vert_{F}^{2}, \
\text{subject to} \ Q(x_i)=F(x_i), \ x_i \in X_{k},
\end{aligned} 
\end{equation}
where $\Vert\cdot\Vert_F$ denotes the Frobenius norm, i.e., for given ${G}\in\mathbb{R}^{m\times n}$, 
\begin{equation*}
\Vert {G}\Vert_F^2=\sum_{i=1}^{m}\sum_{j=1}^{n}\vert {G}_{ij}\vert^2. 
\end{equation*}
 In other words, the quadratic model function $Q_k$ satisfies that $\nabla^{2} Q_k-\nabla^{2} Q_{k-1}$ has the least Frobenius norm over the quadratic functions satisfying the interpolation conditions $Q(x_i)=F(x_i), \ i=1, \ldots, m$. According to the convexity of the Frobenius norm, it is proved that the quadratic model function $Q_k$ is unique at the $k$-th iteration. Thus the rest of coefficients' freedom is taken up. The shared pseudocode of NEWUOA and DFOp is listed in Algorithm \ref{NEWUOA}\footnote{Our solver DFOp, as an improved version, shares the algorithmic framework with NEWUOA.}, where CRVMIN denotes the minimum curvature of the model function. Besides, TRSAPP, BIGLAG and BIGDEN are subroutines in NEWUOA. One can see Powell's paper \cite{powell2006newuoa} for more details of NEWUOA. 
\begin{remark}
Three subroutines TRSAPP, BIGLAG and BIGDEN get the names in the following senses, and the further discussion about them in DFOp is in subsection \ref{Discussion on the quadratic model subproblem}.
\begin{itemize}
\item TRSAPP denotes to find the trust-region subproblem's approximate solution.
\item BIGLAG denotes to obtain a big value of the Lagrange function.
\item BIGDEN denotes to obtain a big denominator of the updating formula.
\end{itemize}
\end{remark}
\begin{algorithm}[H]
\caption{Framework of DFOp\label{NEWUOA}} 
\begin{algorithmic}[1]
\STATE \textbf{Input}:  the black-box objective function $F$ and the initial point $x_0$
\STATE \textbf{Output}: the minimum point $ {x}_{\text{opt}}$ and the minimum $F_{\text{opt}}$
\STATE Initialize and get the interpolation set ${X}$, the initial quadratic model function $Q({x})$, parameters $\rho_{\text{beg}}$, $\Delta$ and $\rho_{\text{end}}$
\STATE Obtain the minimum   $F_{\text{opt}}$ and the corresponding point ${x}_{\text{opt}}$ in the interpolation set
\WHILE{$\rho\ge\rho_{\text{end}}$}
	\STATE solve the trust-region subproblem of the quadratic model 
	$\min_{d}  Q({x}_{\text{opt}}+ d)$,  subject to  $\Vert  d \Vert_{2} \leq \Delta$, 
	 and get $d$ by calling the subroutine TRSAPP,  and then set \text{CRVMIN}
        \STATE set the terminal criterion $=0$
        	\IF{$\Vert  d \Vert_{2} \geq \frac{1}{2} \rho$}
        	   \STATE compute $F( {x}_{\text{opt}}+d)$ and $\text{RATIO} =(F(x_{\text {opt}})-F( {x}_{\text{opt}}+d))/(Q({x}_{\text{opt}})-Q( {x}_{\text{opt}}+d))$; update $\Delta$, and get the index \text{MOVE}
		    \IF {$\text{MOVE}>0$} 
		    \STATE update the interpolation set ${X}$, the model function $Q({x})$ and ${x}_{\text{opt}}$
		   \ENDIF 
              \ELSE\IF{the quadratic model function $Q(x)$ reaches the limit} 
		\STATE       update $\rho$ and $\Delta$,  and set the terminal criterion $=1$
		\ELSE
               \STATE       reduce $\Delta$,  and set $\text{RATIO}=-1$
              \ENDIF 
           \ENDIF 
     \IF{$\text{RATIO}<0.1$ and the terminal criterion $=0$}
     \STATE     find $x_{\text{MOVE}}$, and define $\text{DIST}=\Vert  {x}_{\text{MOVE}}- {x}_{\text{opt}}\Vert_{2}$
        \IF{$\text{DIST} \geq 2 \Delta$} 
         \STATE       call the subroutine BIGLAG and BIGDEN to get $d$ to modify the interpolation model, and update the interpolation set ${X}$, the quadratic model function $Q({x})$ and ${x}_{\text{opt}}$
       \ELSE
         \IF{$\max [\Vert d\Vert_{2}, \Delta] \leq \rho$ and \text{RATIO} $\leq 0$} 
         \STATE update $\rho$ and $\Delta$
        \ENDIF 
       \ENDIF
     \ENDIF
\ENDWHILE 
\IF{$\Vert d\Vert_{2} \leq \frac{1}{2} \rho$}
\STATE calculate $F(x_{\text{opt}}+d)$, and get $x_{\text{opt}}$ and $F_{\text{opt}}$
 \ENDIF
\STATE\RETURN $x_{\text{opt}}$ and $F_{\text{opt}}$
	\end{algorithmic} 
\end{algorithm} 

\subsection{New model updating formula}
It is observed that the %solver NEWUOA and its 
original updating formula of the least Frobenius norm updating quadratic model functions are not suitable for solving problems with transformed/encrypted objective functions, since the corresponding objective function changes with the iteration, which is exactly an important characteristic of the step-transformed\footnote{The term step-transformed/encrypted means that the objective functions' transformations/encryptions depend on the iteration/query step.} problems. Thus we come up with the idea to modify and improve NEWUOA in order to adapt to the black-box optimization with transformed/encrypted objective functions. The new updating formula is supposed to hold the ability to work when the transformed/encrypted objective function $F_k$ changes with the iteration number $k$.

Suppose that the interpolation set at the $k$-th iteration is $\{x_1,\ldots,x_m\}$. The quadratic model of the transformed/encrypted function $F_k$ satisfies that $Q_{\text{old}}(x_i)=F_k(x_i), \ i=1,\ldots,m$, where $m<\frac{1}{2}(n+1)(n+2)$. The index old denotes the $k$-th, and the index new denotes the $k+1$-th, which are both for simplicity and clearness.

What we want is to obtain the new model $Q_{\text{new}}$ satisfying that $Q_{\text{new}}(x_i)=F_{k+1}(x_i)$, $i=1,\ldots, t-1, {\text{new}}, t+1, \ldots, m$, and each $Q_k$ is the solution of 
\begin{equation}\label{leastfrob-2}
\begin{aligned}
\underset{Q}{\operatorname{\min}}\ \left\Vert \nabla^{2} Q-\nabla^{2} Q_{k-1}\right\Vert_{F}^{2}, \
\text{subject to}  \ Q(x_i)=F_k(x_i), \ x_i \in X_{k}.
\end{aligned}
\end{equation}
We omit the iteration index $k$ when there is no ambiguity in this subsection for simplicity. At the current step, we define the quadratic function $D(x)=Q_{\text{new}}(x)-Q_{\text{old}}(x)$. Then $D(x)$ should satisfy the following conditions. 

 If $ i\ne t$, then $D(x_i)=Q_{\text{new}}(x_{i})-Q_{\text{old}}(x_i)=F_{k+1}(x_i)-F_k(x_i)$.

 If $i=t=\text{new}$, denoting $t=\text{new}$ here since $x_t$ is exactly replaced by $x_{\text{new}}$, 
 then $D(x_{\text{new}})=Q_{\text{new}}(x_{\text{new}})-Q_{\text{old}}(x_{\text{new}})=F_{k+1}(x_{\text{new}})-Q_{\text{old}}(x_{\text{new}})$.

In order to obtain $D(x)$, we need to solve the optimization problem 
\begin{equation}
\label{miniFrob}
\begin{aligned}
\underset{D}{\operatorname{\min}} \  \left\Vert \nabla^2 D\right\Vert_F^2, \ 
\text{subject to} \ 
\left\{
\begin{aligned}
D(x_{i})=&F_{k+1}(x_i)-F_k(x_i), \ i=1,\ldots, t-1, t+1, \ldots, m,\\
D(x_{\text{new}})&=F_{k+1}(x_{\text{new}})-Q_{\text{old}}(x_{\text{new}}), \ i=t=\text{new}.
\end{aligned}
\right.
\end{aligned}
\end{equation}
Let $\lambda_j,\ j=1,2,\ldots,m$, be the Lagrange multipliers for the KKT conditions of matrix optimization problem (\ref{miniFrob}), which, as Powell \cite{Leastf} pointed out, have the properties that
\begin{equation}
\sum_{j=1}^{m} \lambda_{j}=0, \ \sum_{j=1}^{m} \lambda_{j}\left(x_{j}-x_{0}\right)=0,\ \nabla^{2} D=\sum_{j=1}^{m} \lambda_{j}\left(x_{j}-x_{0}\right)\left(x_{j}-x_{0}\right)^{\top},\label{3.2}
\end{equation}
where $x_0$ is the base point to reduce the computation errors, and it is initially chosen as the input start point. 
Therefore, the quadratic function $D(x)$ can be written in the form $D(x)=c+(x-x_{0})^{\top} g+\frac{1}{2} \sum_{j=1}^{m} \lambda_{j}((x-x_{0})^{\top}(x_{j}-x_{0}))^{2}$. After determining the parameters $\lambda=(\lambda_1,\ldots,\lambda_m)^{\top} \in \mathbb{R}^m$, $c \in \mathbb{R}$ and $g \in \mathbb{R}^{n}$, we can determine the unique function $D(x)$ and thus obtain the new quadratic model function $Q_{\text{new}}(x)$. Thus the system of linear equations
\begin{equation}
\label{BIG-(3.10)}
\left(\begin{array}{cc}
\mathbf{A} & \mathbf{X}^{\top} \\
\mathbf{X} & \mathbf{0}
\end{array}\right)
\left(
\lambda^{\top},
c,
g^{\top}
\right)^{\top}=\left(\begin{array}{c}
r^{\top}, 
0,\ldots,0
\end{array}\right)^{\top}
\end{equation}
holds, where the matrix $\mathbf{0}\in\mathbb{R}^{(n+1)\times(n+1)}$. The elements of the matrix $\mathbf{A}\in\mathbb{R}^{m\times m}$ and $\mathbf{X}\in\mathbb{R}^{(n+1)\times m}$ are
\begin{equation*}
\begin{aligned}
\mathbf{A}_{i j} =\frac{1}{2}\left(\left(x_{i}-x_{0}\right)^{\top}\left(x_{j}-x_{0}\right)\right)^{2}\ \text{and}\ 
\mathbf{X} =\left(\begin{array}{ccc}
1 & \ldots & 1 \\
x_{1}-x_{0} & \ldots  & x_{m}-x_{0}
\end{array}\right),
\end{aligned}
\end{equation*}
where $1\le i, j \le m$. Besides, the vector $r\in\mathbb{R}^{m}$ has the form as
\begin{equation}
\begin{aligned}
r=&\left(
F_{k+1}(x_1)-F_{k}(x_1),
\ldots,
F_{k+1}(x_{t-1})-F_{k}(x_{t-1}),
F_{k+1}(x_{\text{new}})-Q_{\text{old}}(x_{\text{new}}),\right.\\
&\left.F_{k+1}(x_{t+1})-F_{k}(x_{t+1}),
\ldots,
F_{k+1}(x_m)-F_{k}(x_m)
\right)^{\top}.
\end{aligned}\label{3.3}
\end{equation}
\begin{sloppypar}
Comparing to the vector $r$ in the model updating formula of NEWUOA, which is $(0,\ldots,0,F(x_{\text{new}})-Q_{\text{old}}(x_{\text{new}}),0,\ldots,0)^{\top}$, the vector $r$ in DFOp has a definitely different form. In NEWUOA, there is only one nonzero element in the vector $r$, which is the $t$-th component. However, the elements of the vector $r$ in the new model updating formula of DFOp are nonzero in general cases. The form of the vector $r$ in the new updating formula works well for solving problems with transformed/encrypted objective functions, while NEWUOA's original one can not.
\end{sloppypar} 

It can be observed that if $F_k=F$, for arbitrary  $k\in \mathbb{N}^+$, then the vector $r$ in DFOp will be as same as the vector $r$ in NEWUOA. 

In addition, let the matrix $\mathbf{W}\in\mathbb{R}^{(m+n+1)\times (m+n+1)}$ and the matrix $\mathbf{H}\in\mathbb{R}^{(m+n+1)\times (m+n+1)}$ separately be 
\begin{equation}\label{H}
\mathbf{W}=\left(\begin{array}{cc}
\mathbf{A} & \mathbf{X}^{\top} \\
\mathbf{X} & \mathbf{0}
\end{array}\right)\ 
\text{and}\ \mathbf{H}=\mathbf{W}^{-1}.
\end{equation} 
We can derive the elements of the matrix $\mathbf{W}$ directly from the interpolation points $x_{i}$, $i=1,\ldots, m$. The updating formula of the matrix $\mathbf{H}$ given by Powell \cite{Powell04onupdating} can still be used, which is
\begin{equation}
\begin{aligned} 
\mathbf{H}_\text{new}=& \mathbf{H}+\sigma^{-1}\left(\alpha\left(e_{t}-\mathbf{H} w\right)\left(e_{t}-\mathbf{H} w\right)^{\top}-\beta \mathbf{H} e_{t} e_{t}^{\top} \mathbf{H}\right.\\
+&\left. \tau\left(\mathbf{H} e_{t}\left(e_{t}-\mathbf{H} w\right)^{\top}+\left(e_{t}-\mathbf{H} w\right) e_{t}^{\top} \mathbf{H} \right)\right), 
\end{aligned}
\label{3.4}
\end{equation}
where 
\begin{equation}
\label{BIG-(4.12)}
\begin{aligned} 
\alpha= e_{t}^{\top}\mathbf{H}e_{t}, \  \beta=\frac{1}{2}\left\Vert x_{\text{new}}-x_{0}\right\Vert ^{4}_2-w^{\top}\mathbf{H}w, \ 
\tau=e_{t}^{\top}\mathbf{H}w,  \ \sigma=\alpha \beta+\tau^{2}, 
\end{aligned}
\end{equation}
and the vector $w \in \mathbb{R}^{m+n+1}$ is defined as
\begin{equation*}
\left\{\begin{array}{ll}w_{i}=\frac{1}{2}\left(\left({x}_{i}-{x}_{0}\right)^{\top}\left({x}_{\text{new}}-{x}_{0}\right)\right)^{2}, &\ i=1,2, \ldots, m, \\ 
w_{m+1}=1 \text { and } w_{i+m+1}=\left({x}_{\text{new}}-{x}_{0}\right)_{i}, &\ i=1,2, \ldots, n. \end{array}\right. 
\end{equation*}
Besides, $e_t\in \mathbb{R}^{m+n+1}$ denotes the vector of which the $t$-th component is $1$ and the other components are all $0$. 
Finally, the updating formula for $D(x)$ is
\begin{equation*}
D(x)=c+\left(x-x_{0}\right)^{\top} g+\frac{1}{2} \sum_{j=1}^{m} \lambda_{j}\left(\left(x-x_{0}\right)^{\top}\left(x_{j}-x_{0}\right)\right)^{2}, \ x \in \mathbb{R}^{n},
\end{equation*}
and the parameters $\lambda$, $c$ and $g$ are given by 
\begin{equation*}
(
\lambda^{\top},
c,
g^{\top}
)^{\top}
=\mathbf{H}(
r^{\top},0,\ldots,0
)^{\top},
\end{equation*}
where $r$ has the form as (\ref{3.3}), and the matrix $\mathbf{H}$ is updated to $\mathbf{H}_{\text{new}}$ as formula (\ref{3.4}).

\subsection{Simplification details and interpolation process}

By the virtue of $F_k(x_{\text{opt}})=Q_{\text{old}}(x_{\text{opt}})$, where $x_{\text{opt}}$ denotes the minimum point among the $k$-th interpolation points, the vector $r=r-(0,\ldots,0)^{\top}$ can be reformulated as
\begin{equation*}
\begin{aligned}
r=&\left(
F_{k+1}(x_1)-F_k(x_1),
\ldots,
F_{k+1}(x_{t-1})-F_k(x_{t-1}),
\delta_k,\right.\\
&\left.F_{k+1}(x_{t+1})-F_k(x_{t+1}),
\ldots,
F_{k+1}(x_m)-F_k(x_m)
\right)^{\top},
\end{aligned}
\end{equation*}
in which, the parameter $\delta_k$ is defined as
\begin{equation}\label{delta_k}
\delta_k= \left(F_{k+1}(x_{\text{new}})-F_k(x_{\text{opt}})\right)-\left(Q_{\text{old}}(x_{\text{new}})-Q_{\text{old}}(x_{\text{opt}})\right).
\end{equation}

The parameter $\delta_k$ can be regarded as the difference between the change of objective function values and the change of the old model function values from $x_{\text{opt}}$ to $x_{\text{new}}$. The new parameter $\delta_k$ is used in the implementation of DFOp, which aims to reduce the computation error and avoid needless computation of the constant term $c$.
 
%%%%%%%%%%%%%%%%%%%%%%%%%
 
In the implementation, we should take careful notice of the updating procedure. When solving an $n$-dimension problem with DFOp, we synchronize the updating of the first $m$ interpolation points. The iteration points are not used directly to obtain the model functions in the first $m-1$ iterations.
\begin{table}[htbp] 
\fontsize{8}{8}\selectfont
	\centering   
	\caption{Interpolation in DFOp for solving problems with transformed/encrypted objective functions\label{table1}}  
\fontsize{8}{8}\selectfont
	\setlength{\abovecaptionskip}{0.cm}
	\begin{tabular}{llllll}  
		\hline Step  & Set & \multicolumn{4}{c}{Function evaluation} \\   
		\hline $1$ & $\{x_1\}$  & $F_1(x_1)$&-&-&-\\    		
		$2$ & $\{x_1, x_2\}$  & $F_2(x_1)$ & $F_2(x_2)$&- &-\\
		\vdots & \vdots &  $\vdots$ &  $\vdots$&  $\vdots$& $\vdots$\\
		$m-1$  & $\{x_1, \ldots, x_{m-1}\}$  & $ F_{m-1}(x_{1})$ & $\ldots$ & $F_{m-1}(x_{m-1})$&-\\     
		$m$ & $\{x_1, \ldots, x_{m}\}$   & $ F_m(x_{1})$ & $\ldots$ &$F_m(x_{m-1})$& $F_m(x_{m})$\\
		$m+1$  & $\{x_2, \ldots, x_{m+1}\}$   & $F_{m+1}(x_{2})$ & $\ldots$&$F_{m+1}(x_{m})$ & $F_{m+1}(x_{m+1})$\\
		\vdots & \vdots &  $\vdots$ &  $\vdots$&  $\vdots$&  $\vdots$\\
		$k$ & $\{x_{k-m+1}, \ldots, x_{k}\}$ & $ F_k(x_{k-m+1})$ & $\ldots$ & $F_k(x_{k-1})$& $F_k(x_{k})$\\
\hline                                    
	\end{tabular}
\end{table}
Once the function evaluation at the new iteration point is finished, and the previous output function values are synchronized, the points in the interpolation set are updated at each step according to the process shown in Table \ref{table1}. In this way, it is ensured that the transformed/encrypted function values corresponding to each interpolation set are with the uniform transformation. 

For simplicity, we denote the dropped interpolation point with bad properties as $x_{k-m}$ at the $k$-th step in Table \ref{table1}. It may be a different point in practice. The purpose here is to show the interpolation updating process and to be consistent with the real updating process of adding a point and then dropping a point in DFOp. The cost of function evaluations will not increase in practice, which will be further discussed later in Remark \ref{callingtime} for applications to private black-box problems.  

\subsection{Discussions on the quadratic model subproblem}
\label{Discussion on the quadratic model subproblem}

One of the termination conditions of the subroutine for solving the quadratic model subproblem in the trust-region, which is called TRSAPP in (NEWUOA and) DFOp, is
$
\Vert d_k \Vert < \frac{1}{2} \rho_{k},
$
where $d_k$ is the solution of the trust-region subproblem
\begin{equation}\label{3.5}
\begin{aligned}
\underset{{d \in \mathbb{R}^n}}{\operatorname{\min}} \ Q_k(x_{\text{opt}}+d),\ 
\text{subject to}  \ \Vert d \Vert_{2} \le \Delta_{k}.
\end{aligned}
\end{equation}
The parameter $\rho_{k}$ is an iterative lower bound on the trust-region radius, and it is designed to maintain enough distance between the interpolation points. Besides, $x_{\text{opt}}$ denotes the minimum point at the $k$-th iteration.

A judgment will be made in the subroutine TRSAPP. If $\Vert d_k \Vert_{2} \ge \frac{1}{2} \rho_k$, then it calculates $F_k({x}_{\text{opt}}+{d}_k)$, and sets the parameter $\text{RATIO}=(F_k({x}_{\text{opt}})-F_k({x}_{\text{opt}}+{d_k}))/(Q_k({x}_{\text{opt}})-Q_k({x}_{\text{opt}}+{d_k}))$. Afterwards, the radius of the trust-region $\Delta_k$ is revised\footnote{Details of the revision are in Powell's paper \cite{powell2006newuoa}.} according to the parameter RATIO, subject to $\Delta_k \ge \rho_{k}$. Then the index $\text{MOVE}$\footnote{The index MOVE is designed to denote the interpolation point that will be dropped.} is set to $0$ or the index of the interpolation point that will be dropped at the next step. If $\Vert d_k \Vert_{2} < \frac{1}{2} \rho_{k}$, then it will test whether three recent values of $\Vert d_k \Vert_{2}$ and $\vert F_k({x}_{\text{opt}}+{d_k})-Q_k({x}_{\text{opt}}+{d_k})\vert$ are small enough compared to the parameter $\rm{CRVMIN}$, and $\rm{CRVMIN}$ here denotes the minimum of the curvature, i.e., $\min_{d \in \mathbb{R}^n} \frac{d^{\top} \nabla^2 Q_k d}{d^{\top} d}$. Therefore, the termination in TRSAPP has an important role, by which the number of function evaluations is skillfully deflected.

For the black-box optimization problem without transformations, the objective function does not change, and only the trust-region changes when the iteration increases. However, in the black-box optimization with transformed/encrypted objective functions we consider, $F_k$ and the trust-region both change when the iteration increases. The change of the function $F_k$ leads to the change of the quadratic model $Q_k$. Then the solution of the subproblem $d_k$ may change according to its definition. As a result, the condition $\Vert d_k \Vert_{2} < \frac{1}{2} \rho_{k}$ may not hold at all, and the termination will be influenced.

If the termination condition is not satisfied, the iteration can not leave the loop of TRSAPP. Firstly, the number of the iterations perhaps grows up to the infinity, which refers to a high cost for function evaluations, but still can not reach the wanted approximate minimum of $F$. Secondly, the model-improvement step given by BIGLAG and BIGDEN can hardly be called, which can not reduce the interpolation error. 

Notice that what our solver DFOp expects is a relatively large modulus of the denominator $\sigma=\alpha\beta+\tau^2$ according to (\ref{3.4}) and (\ref{BIG-(4.12)}). Usually $\sigma>\tau^2$ holds in practice, since in theory both $\alpha$ and $\beta$ should be positive, and we can use $\vert\alpha\vert\vert\beta\vert$ instead of $\alpha\beta$ in the implementation. Thus the subroutine BIGLAG obtains $d_k$ by solving the subproblem
\begin{equation*}
\max_d\ \left\vert\ell_{t}({x}_{\mathrm{opt}}+ {d})\right\vert, \ 
\text{subject to} \ \Vert  {d}\Vert_{2} \leq {\Delta}_{k},
\end{equation*}
since $\tau=\ell_{t}({x}_{\mathrm{opt}}+ {d})$, where $\ell_{t}$ is the $t$-th Lagrange function. In addition, the subroutine BIGDEN seeks a large value of the denominator's absolute value $\vert\sigma\vert$ in (\ref{3.4}), and thus it obtains $d_k$ by solving the subproblem
\begin{equation*}
\max_d\ \left\vert\sigma( {x}_{\mathrm{opt}}+ {d})\right\vert, \ 
\text{subject to} \ \Vert  {d}\Vert_{2} \leq {\Delta}_{k}.
\end{equation*}

The aim to reduce the singularity of the system that defines the Lagrange functions can still be achieved by calling BIGLAG and BIGDEN in DFOp. However, the iteration only implements the criticality step by calling TRSAPP instead of improving the quadratic model function, since BIGLAG and BIGDEN are rarely called in the iteration. The denominator $\sigma$ in (\ref{3.4}) might be close to 0, which will make the singularity of the system increase. Thus the error in computing the quadratic model functions may become larger, which makes DFOp invalid.

\subsection{The solution shift}

We should figure out in which condition the change of the transformed/encrypted function $F_k$ does not influence the convergence of DFOp. It can be observed that the key of the termination condition is the value of $\Vert d_k \Vert_{2}$. If it is not changed dissolvingly by the transformation, then the transformed/encrypted function $F_k$ at each iteration does not influence the convergence.
\begin{definition}
We define the model solution shift $\mathcal{D}_m$ from the function $F_{k-1}$ to the function $F_{k}$ as
$
\mathcal{D}_{m}(F_{k-1},F_k) =\Vert  d_{k}\Vert_{2},
$
where $d_{k}$ is the solution of subproblem (\ref{3.5}). 
\end{definition}
The model solution shift holds the name for two aspects. Firstly, $\Vert d_k \Vert_{2}$ here denotes the norm of the solution of a model subproblem in the trust-region, which also denotes the distance between two solutions of the neighboring subproblems\footnote{We discuss the most common case that the minimum points at most iterations are given by TRSAPP.}. Secondly, $\Vert d_k \Vert_{2}$ depends on the quadratic model function $Q_{k}$, which will change when the objective function $F_k$ changes. 

The solution $d_{k}$ of the model subproblem is an approximation to the solution of the subproblem
$
\min_{d} \ F_{k} (x_{\text{opt}}+d),\ 
\text{subject to}  \ \Vert d\Vert_{2} \leq \Delta_{k},
$
since $Q_{k}$ is the local interpolation quadratic model function of $F_{k}$, and it is proved to be a fully linear model function. Besides, the judgment based on the parameter $\text{RATIO}=(F_{k}(x_{\text{opt}})-F_{k}(x_{\text{opt}}+d_{k}))/(Q_{k}(x_{\text{opt}})-Q_{k}(x_{\text{opt}}+d_{k}))$ can ensure the decrease of $F_{k}$ at the point that  DFOp will accept, which can be regarded as a guarantee of the interpolation as well, since a good interpolation point's corresponding RATIO is close to 1. Besides, we propose the definition of the solution shift $\mathcal{D}_{s}$ in the following. Both the model solution shift $\mathcal{D}_{m}$ and  the solution shift $\mathcal{D}_{s}$ illustrate the influence caused by the transformation, and will help us analyze the convergence property of DFOp.
\begin{definition}
We define the solution shift $\mathcal{D}_s$ from the function $F_{k-1}$ to the function $F_{k}$ as
\begin{equation*}
\mathcal{D}_{s}(F_{k-1},F_k) =\left\{
\begin{array}{cl}
\underset{x^*_{i,k-1}\in \mathcal{LM}_{k-1},\ x^*_{i,k}\in \mathcal{LM}_{k}} {\sum}  \left\Vert  x^*_{i,k}-x^*_{i,k-1}\right\Vert_{2}, 
&\text{if $\vert\mathcal{LM}_{k-1}\vert=\vert\mathcal{LM}_{k}\vert$},\\
 +\infty,
 & \text{otherwise},
\end{array}
\right.
\end{equation*}
where $\mathcal{LM}_k$ is the ordered set of all local minimum points of the corresponding $k$-th objective function $F_k$, and $\vert\cdot\vert$ denotes the cardinality.
\end{definition}

When $F_k \ne F_{k-1}$, the solution shift may become larger than that of the situation where $F_k = F_{k-1}$. If $\mathcal{D}_{s}(F_{k-1},F_k)=\mathcal{D}_{s}(F, F)$, the transformation from $F$ to $F_k$ at the $k$-th iteration, for each $k\in\mathbb{N}^+$, has no dissolving influence on the convergence of DFOp. Moreover, if the model solution shift $\mathcal{D}_{m}(F_{k-1},F_k)=\mathcal{D}_{m}(F, F)$ holds, we can get stronger results.

\section{Transformations/encryptions and the Solution Shift}
\label{Linear transformation and the model function}

This section will illustrate the objective functions with special transformations/encryptions and the model functions given by DFOp. We will also discuss more other transformations/encryptions generally. The analysis is the basis for the convergence analysis.

\subsection{Positive monotonic transformations/encryptions and examples}

We firstly present the definition of positive monotonic transformations.
\begin{definition}\label{PMT}
If a transformation $T_k: \mathbb{R}\rightarrow \mathbb{R}$ preserves the order of data, namely $T_k(x)>T_k(y)$ for $x>y$, then we say that $T_k$ is a positive monotonic transformation.
\end{definition}

\begin{sloppypar} 
Then we can directly get the following theorem about the solution shift.
\begin{theorem}\label{monotone}
If the transformation on the objective function $F$ at the $k$-th step, for each $k\in\mathbb{N}^+$, defined as $T_{k}$, is a positive monotonic transformation, then it holds that 
$\mathcal{D}_s(T_{k-1}(F),T_{k}(F))=\mathcal{D}_s(F,F).$
\end{theorem}
\begin{proof}
The conclusion is derived from Definition \ref{PMT}, since the transformed/encrypted function $T_{k}(F)$ shares the same corresponding minimal points with the original objective function $F$ for each $k\in\mathbb{N}^+$. 
\end{proof}
\end{sloppypar}

\begin{sloppypar} 
Positive monotonic transformations/encryptions can be any strictly monotonic increasing functions, such as linear functions, exponential functions and power functions with odd positive power. 
We give the direct example of a kind of linear transformations/encryptions with details here.
\begin{example}
We denote the linearly transformed/encrypted objective function as $F_k=C_{1,k}F+C_{2,k}$, where the constants $C_{1,k}, C_{2,k}\in \mathbb{R}$, and $C_{1,k}>0$, for arbitrary $k \in \mathbb{N}^+$. Then it holds that $\mathcal{D}_s(F_{k-1},F_k)=\mathcal{D}_s(F,F).$
\end{example}
\end{sloppypar}

The example above involves multiplication and addition, which are useful for the application to private black-box optimization problems. To study the model solution shift $\mathcal{D}_m$, we firstly derive the following theorem of the simplest example, the translation transformation.
\begin{theorem}\label{corollary5.4-xpc-2-test}
Suppose that the function $F:\mathbb{R}^{n} \rightarrow \mathbb{R}$ is the original objective function, and $Q_{k}$ is its quadratic model function at the $k$-th iteration given by NEWUOA when solving the problem without transformations. If the objective function $F$ is transformed/encrypted to $F_{k}=F+C_{2,k} $ at the $k$-th iteration, for each $k \in \mathbb{N}^{+}$, then its model function given by DFOp, denoted by $\bar{Q}_{k}$ at the $k$-th iteration, satisfies that $\bar{Q}_{k}=Q_{k}+C_{2,k}$, and thus $\mathcal{D}_{m}(F_{k-1},F_k)=\mathcal{D}_{m}(F,F)$.
\end{theorem}
\begin{proof}
We denote the initial interpolation sets of $F$ and $F+C_{2,1}$ as $X_1$ and $\bar{X}_1$ respectively. For a certain problem, it holds that $X_1=\bar{X}_1$. Besides, we denote the initial quadratic models of $F$ and $F+C_{2,1}$ as $Q_1$ and $\bar{Q}_1$ separately. Since there is no difference of both initial interpolation sets and the finite difference methods used to form the initial quadratic model, it holds that
$\bar{Q}_1={Q}_1+C_{2,1}$.  

Suppose that $\bar{Q}_{k-1}={Q}_{k-1} +C_{2,k-1}$ and $X_{k-1}=\bar{X}_{k-1}$, where $X_{k-1}$ is the interpolation set of ${Q}_{k-1}$,  and $\bar{X}_{k-1}$ is the interpolation set of $\bar{Q}_{k-1}$. Then the new interpolation points $x_{\text{new}}' \in X_k$ and $\bar{x}_{\text{new}}' \in \bar{X}_{k}$ satisfy that $x_{\text{new}}'=\bar{x}_{\text{new}}',$ and the dropped interpolation points $x_{t}' \in X_{k-1}$ and $\bar{x}_{t}' \in \bar{X}_{k-1}$ satisfy that $x_{t}'=\bar{x}_{t}'$. Since $X_k=X_{k-1} \cup \{x_{\text{new}}'\}\backslash\{x_t'\}\ \text{and}\ \bar{X}_k=\bar{X}_{k-1} \cup \{\bar{x}_{\text{new}}'\}\backslash\{\bar{x}'_t\},$ we have $X_{k}=\bar{X}_{k}$. 

According to the interpolation conditions of NEWUOA and DFOp, we know that the quadratic model function $\bar{Q}_k$ is the unique solution of the problem
\begin{equation*}
\min_{Q} \ \Vert \nabla^2Q-\nabla^2 \bar{Q}_{k-1} \Vert_F^2,\ 
\text{subject to} \ Q(x_i)=F(x_i)+C_{2,k} , \ x_i \in \bar{X}_k.
\end{equation*}
In addition, according to the definition, the quadratic function $Q_k$ is the unique solution of problem (\ref{leastfrob}). Hence the quadratic function ${Q}_{k}+C_{2,k-1}$ is the unique solution of
\begin{equation*}
\min_{Q} \ \Vert \nabla^2 Q-{\nabla^2 Q}_{k-1} \Vert_F^2, \ 
\text{subject to} \ {Q}(x_i)=F(x_i)+C_{2,k-1} , \ x_i \in X_k.
\end{equation*}
Thus the function ${Q}_{k}(x)+C_{2,k}$ satisfies that $Q_k(x_i)+C_{2,k}=F(x_i)+C_{2,k}, \ x_i \in X_k$, and is naturally the unique solution of 
\begin{equation*}
\min_{Q} \ \Vert \nabla^2 Q-\nabla^2\bar{Q}_{k-1} \Vert_F^2, \ 
\text{subject to} \ {Q}(x_i)=F(x_i)+C_{2,k} , \ x_i \in \bar{X}_k.
\end{equation*} 
Therefore, $\bar{Q}_{k}=Q_{k}+C_{2,k}$, $\forall k\in \mathbb{N}^{+}$. Consequently $\mathcal{D}_{m}(F_{k-1},F_k)=\mathcal{D}_{m}(F,F)$ holds. 
\end{proof}

\subsection{Transformations with no change of the model solution shift}
\label{Transformation maintaining model solution shift}

A natural question is whether there is any other transformation that can even keep the model solution shift do not change at each iteration, in addition to the translation transformation. Moreover, we can try to achieve such transformation if the answer is yes, since stronger analyses can be given for the case where there is even no change on the model solution shift. To answer the question above, we propose the following definition.
 
\begin{definition}
\label{N1}
Given $m$ known interpolation points $x_1,\ldots,x_m$ at the $k$-th iteration, containing the current minimum point $x_{\text{opt}}$, $\mathcal{Q}$ is the mapping from the $m$ function values at the $m$ points to the model $Q_k$ by solving problem (\ref{leastfrob-2}).  
\end{definition}

\begin{definition}\label{N1-2}
The operator $\operatorname{argtrmin}$ denotes getting the minimum point of $Q_k$ in the corresponding trust-region by solving problem (\ref{3.5}) at the $k$-th iteration.
\end{definition}

\begin{definition}\label{N1-3}
The mapping $\mathcal{N}=\operatorname{argtrmin}\circ\mathcal{Q}$ is named as the model minimum mapping, which illustrates the interpolation process and the subroutine TRSAPP in our DFOp for solving problem (\ref{3.5}). The mapping $\mathcal{N}^{-1}$ is the inverse mapping of $\mathcal{N}$.  
\end{definition}

\begin{sloppypar} 
Fig.\ref{communicative diagram} is a commutative diagram, in which $\mathcal{N}(\mathcal{F})=\mathcal{N}(\check{\mathcal{T}}_k(\mathcal{F}))$ means that the corresponding model functions separately based on the original function values and transformed/encrypted function values at the interpolation points given by DFOp share the same minimum point or the solution of the trust-region subproblem. $\mathcal{F}$ denotes the vector $(F(x_1),\ldots,F(x_m))^{\top}$, and $\check{\mathcal{T}}_k(\mathcal{F})$ denotes the vector $(\check{T}_k(F(x_1)),\ldots,\check{T}_k(F(x_m)))^{\top}$, where $x_i \in X_k, i=1,\ldots, m$, since for the black-box problem, $\mathcal{F}=(F(x_1),\ldots,F(x_m))^{\top}$ represents all characteristics of the objective function $F$. Besides, $\check{\mathcal{T}}_k$ denotes restricting the transformation $\check{T}_k$ in the $m$ function values $F(x_1),\ldots,F(x_m)$, and $\check{T}_k$ here is the transformation with no change of the model solution shift. 
\begin{figure}[hbtp]
\centering
\adjustbox{scale=1,center}{%
\begin{tikzcd}
\mathcal{F}\in \mathbb{R}^m \arrow[r, "\check{\mathcal{T}}_k"] \arrow[dd, bend right=90, "\mathcal{N}"swap] \arrow[d, "\mathcal{Q}"]
& \check{\mathcal{T}}_k(\mathcal{F}) \in\mathbb{R}^m \arrow[dd, bend left=90, "\mathcal{N}"] \arrow[d, "\mathcal{Q}"swap] \\
\begin{pmatrix}
%\left(
\lambda_F\\
c_F\\
g_F
\end{pmatrix}
%\right)^{\top}
\in\mathbb{R}^{m+n+1}\arrow[d, "\operatorname{argtrmin}"]  
& 
\begin{pmatrix}
%\left(
\lambda_T\\
c_T\\
g_T
\end{pmatrix}
%\right)^{\top}
\in\mathbb{R}^{m+n+1}\arrow[d, "\operatorname{argtrmin}"swap] 
\\
{d}^*_k\in\mathbb{R}^{n}  \arrow[r, equal] 
& 
{d}^*_k\in\mathbb{R}^{n}
\end{tikzcd}
}
\caption{Commutative diagram: $Q_k$ is obtained based on the known $Q_{k-1}$ and $(\lambda_F^{\top},c_F,g_F^{\top})^{\top}$ or $(\lambda_T^{\top},c_T,g_T^{\top})^{\top}$, and thus $\mathcal{Q}(\mathcal{F})$ and $\mathcal{Q}(\check{\mathcal{T}}_k(\mathcal{F}))$ are denoted by $(\lambda_F^{\top},c_F,g_F^{\top})^{\top}$ and $(\lambda_T^{\top},c_T,g_T^{\top})^{\top}$ for clearness. \label{communicative diagram}}
\end{figure}
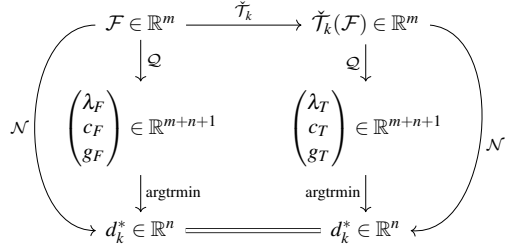
\end{sloppypar} 
%}

We can derive the following theorem after giving the definition above.
\begin{theorem}
\label{Tao}
Given $d^*_k\in \mathbb{R}^n$, $\mathcal{N}^{-1}({d}^*_k)$ contains a translated linear space, of which the dimension is at least $m-n$, where $d^*_k$ denotes the solution of the quadratic model subproblem in the trust-region $\mathbb{B}_{\Delta_k}({x_{\text{opt}}})$. Then the transformation $\check{\mathcal{T}}_k: \mathcal{F}\rightarrow \check{\mathcal{T}}_k(\mathcal{F})$, satisfying that $\mathcal{N}(\mathcal{F})=\mathcal{N}(\check{\mathcal{T}}_k(\mathcal{F}))$ and then $\mathcal{D}_{m}(F_{k-1},F_k)=\mathcal{D}_{m}(F,F)$ at the $k$-th iteration, can be other transformations/encryptions in addition to the translation transformation referring to ${F}+C_{2,k}$, where $F_k=\check{T}_k(F)$. 
\end{theorem}
\begin{proof}
According to Definition \ref{N1}, $\mathcal{Q}$ aims to obtain the $k$-th quadratic model function $\check{Q}_k(x)$ from $\nabla^2\check{Q}_k$, $\check{c}_k$, and $\check{g}_k$.  Denote $\check{D}_{k}=\check{Q}_{k}-\check{Q}_{k-1}$ with the parameters $\lambda^{*}_k\in\mathbb{R}^{m}, c^{*}_k\in\mathbb{R}, g^{*}_k\in\mathbb{R}^{n}$. Then it holds that
\begin{equation*}\label{LS}
\begin{aligned}
\left(
({\lambda}^{*}_k)^{\top},
{c}^{*}_k,
({g}^{*}_k)^{\top}
\right)^{\top}=&\mathbf{H}_k\left(
\check{T}_k(F(x_{1})) - \check{Q}_{k-1}(x_{1}),
\ldots,
\check{T}_k(F(x_{m}))  - \check{Q}_{k-1}(x_{m}),
0,
\ldots,
0
\right)^{\top}.%\\
%+&
%\left(
%\check{\lambda}_{k-1},
%\check{c}_{k-1},
%\check{g}_{k-1}
%\right)^{\top},
\end{aligned}
\end{equation*}
%where the constant parameters $\check{\lambda}_{k-1}\in\mathbb{R}^{m}, \check{c}_{k-1}\in\mathbb{R}, \check{g}_{k-1}\in\mathbb{R}^{n}$ depend on the parameters of the $k-1$-th quadratic model at the $k-1$-th iteration, $\check{Q}_{k-1}$ say, but not on $(\check{T}_k(F(x_1)),\ldots,\check{T}_k(F(x_m)))^{\top}$ at the $k$-th step. 
When $\check{T}_k(F)=F$, which is the original objective function,  the solution of the trust-region subproblem (\ref{3.5}), $d^*_k$ say, can be obtained. Then we are going to obtain the space $\mathcal{N}^{-1}(d^*_k)$ for the known $d^*_k$, which also means to find other transformations/encryptions $\check{T}_k$ and their corresponding $\check{\mathcal{T}}_k$ that satisfy $\mathcal{N}(\check{\mathcal{T}}_k(\mathcal{F}))=\mathcal{N}(\mathcal{F})$. According to the KKT conditions, $d^*_k$ is the solution of the trust-region subproblem (\ref{3.5}), if and only if $d^*_k$ is feasible and there exists $\omega\ge 0$ satisfying that
\begin{align}
(\nabla^{2} \check{Q}_{k}+\omega \mathbf{I}) d^{*}_k&=-\check{g}_k,\label{1-1-1}\\
\omega(\Delta_k-\Vert d^*_k\Vert_2)&=0,\label{1-1-3}\\
\nabla^{2} \check{Q}_{k}+\omega \mathbf{I} &\succeq \mathbf{0},\label{1-1-4}
\end{align}
where $\mathbf{I} \in\mathbb{R}^{n\times n}$ is the identity matrix and (\ref{1-1-4}) denotes that $\nabla^{2} \check{Q}_{k}+\omega \mathbf{I}$ is positive semidefinite. %The given solution $d^{*}$ of the trust-region subproblem is calculated by the operator $\mathcal{N}$ based on $\vec{F} \in \mathcal{N}^{-1}(d^{*})$. 

If $\Vert d^{*}_k\Vert_{2} < \Delta_k$, then $\omega=0$ holds according to condition (\ref{1-1-3}). If $\nabla^{2} \check{Q}_{k}$ is positive semidefinite,  then condition (\ref{1-1-4}) holds naturally. We omit the other situations of $\nabla^{2} \check{Q}_{k}$, considering that we can always choose $\check{\mathcal{T}}_k(\mathcal{F})$ of which the corresponding model function has positive semidefinite Hessian. We denote the $j$-th row of the matrix $\mathbf{H}_{k}$ as $(\mathbf{H}_{k})_{j}$. Thus we can obtain $\check{\mathcal{T}}_k(\mathcal{F})=(\check{T}_k(F(x_1)),\ldots,\check{T}_k(F(x_m)))^{\top}$ by solving the linear equations of $\check{\mathcal{T}}_k(\mathcal{F})$, which are
%
%Thus we can obtain $\vec{F}$ by solving 
\begin{equation}\label{1-2}
\begin{aligned}
&\sum_{j=1}^{m} \left(\left(x_{j}-x_{0}\right)\left(x_{j}-x_{0}\right)^{\top}d^*_k\right) \bigg((\mathbf{H}_k)_j\left(\check{T}_k({F}(x_{1})) - \check{Q}_{k-1}(x_{1}),\ldots,\right.\\
&\left. \check{T}_k({F}(x_{m})) - \check{Q}_{k-1}(x_{m}),0,\ldots, 0\right)^{\top}\bigg)+\nabla^2\check{Q}_{k-1}d^*_k\\
=&-\left((\mathbf{H}_k)^{\top}_{m+2},\ldots,(\mathbf{H}_k)^{\top}_{m+n+1}\right)^{\top}\bigg(\check{T}_k({F}(x_{1})) - \check{Q}_{k-1}(x_{1}),\ldots, \\
&\check{T}_k({F}(x_{1})) - \check{Q}_{k-1}(x_{m}),0,\ldots, 0\bigg)^{\top}-\check{g}_{k-1},
\end{aligned}
\end{equation}
since $\nabla^{2} \check{Q}_{k} d^{*}_k=-\check{g}_k$. 

If $\Vert d^{*}_k\Vert_{2} = \Delta_k$, then we can always choose a constant $\omega$ to achieve $\nabla^{2} \check{Q}_{k}+\omega \mathbf{I} \succeq \mathbf{0}$, and then we can obtain $\check{\mathcal{T}}_k(\mathcal{F})=(\check{T}_k(F(x_1)),\ldots,\check{T}_k(F(x_m)))^{\top}$  
by solving the linear equations of $\check{\mathcal{T}}_k(\mathcal{F})$, which are
 
\begin{equation}\label{1-1}
\begin{aligned}
&\sum_{j=1}^{m} \left(\left(x_{j}-x_{0}\right)\left(x_{j}-x_{0}\right)^{\top}d^*_k\right) \bigg((\mathbf{H}_k)_j\left(\check{T}_k({F}(x_{1})) - \check{Q}_{k-1}(x_{1}),\ldots,\right.\\
&\left. \check{T}_k({F}(x_{m})) - \check{Q}_{k-1}(x_{m}),0,\ldots, 0\right)^{\top}\bigg)+\nabla^2\check{Q}_{k-1}d^*_k+\omega \mathbf{I} d^{*}_k\\
=&-\left((\mathbf{H}_k)^{\top}_{m+2},\ldots,(\mathbf{H}_k)^{\top}_{m+n+1}\right)^{\top}\bigg(\check{T}_k({F}(x_{1})) - \check{Q}_{k-1}(x_{1}),\ldots, \\
&\check{T}_k({F}(x_{m})) - \check{Q}_{k-1}(x_{m}),0,\ldots, 0\bigg)^{\top}-\check{g}_{k-1},
\end{aligned}
\end{equation}
since $(\nabla^{2} \check{Q}_{k}+\omega \mathbf{I}) d^{*}_k=-\check{g}_k$. 

Hence, $\mathcal{N}^{-1}\left(d^{*}_k\right)$, as the solution space of linear system (\ref{1-2}) or (\ref{1-1}), contains a translated linear space whose dimension is at least $m-n$, since the number of unknown elements of $\check{\mathcal{T}}_k(\mathcal{F})=(\check{T}_k(F(x_{1})), \ldots, \check{T}_k(F(x_{m})))^{\top}$ is $m$ and there are $n$ equations in the linear system. 
Then we can derive that the transformation
$
\check{\mathcal{T}}_k: \mathcal{N}^{-1}({d}^*_k) \rightarrow \mathcal{N}^{-1}({d}^*_k)
$
is not unique, since the dimension of
$\mathcal{N}^{-1}({d}^*_k)$ is at least $m-n \geq 2$, and the restricted transformation $\check{\mathcal{T}}_k$ includes nonlinear or linear mapping. In addition, $(\check{T}_k ({F}(x_{1})), \ldots, \check{T}_k ({F}(x_{m})))^{\top}$ can be calculated by solving linear systems (\ref{1-2}) and (\ref{1-1}). The corresponding transformed/encrypted functions satisfy that $\mathcal{D}_{m}(F_{k-1},F_{k})=\mathcal{D}_{m}(F,F)$, if all of $F_k$, for each $k\in \mathbb{N}^+$, are transformed/encrypted in the way $F_k=\check{T}_k(F)$. The proof of Theorem \ref{Tao} is completed.
\end{proof}

Therefore, in addition to the translation transformation $F+C_{2,k}$, there are indeed more other transformations $\check{T}_k$ that satisfy $\mathcal{N}(\check{\mathcal{T}}_k(\mathcal{F}))=\mathcal{N}(\mathcal{F})$ and then $\mathcal{D}_m(F_{k-1},F_k)=\mathcal{D}_m(F,F)$, where $\check{\mathcal{T}}_k$ is the restriction of $\check{T}_k$.

\begin{remark}
According to the proof of Theorem \ref{Tao}, the transformation $F+C_{2,k+1}$ at each step keeps the model solution shift, since
\begin{equation}\label{otherwayofC+}
\begin{aligned} 
\left(
\lambda^{\top},
c+C_{2,k+1},
g^{\top}
\right)^{\top}=&\mathbf{H}_{k+1} \left(C_{2,k+1},\ldots,C_{2,k+1},F(x_{\text{new}})+C_{2,k+1}-Q_{\text{old}}(x_{\text{new}}),\right.\\
&\left.C_{2,k+1},\ldots,C_{2,k+1},0,\ldots,0\right)^{\top},
\end{aligned}
\end{equation}
where $F(x_{\text{new}})+C_{2,k+1}-Q_{\text{old}}(x_{\text{new}})$ is the $t$-th component,  and the number of $0$ is $n+1$, if
\begin{equation*}
\begin{aligned} 
\left(
\lambda^{\top},
c,
g^{\top}
\right)^{\top}=&\mathbf{H}_{k+1} \left(0,\ldots,0,F(x_{\text{new}})-Q_{\text{old}}(x_{\text{new}}),0,\ldots,0\right)^{\top}.
\end{aligned}
\end{equation*}
(\ref{otherwayofC+}) is based on the fact that 
\begin{equation*}
\mathbf{W}_{k+1}
\left(0,\ldots,0,C_{2,k+1},0,\ldots,0\right)^{\top}=\left(C_{2,k+1},\ldots,C_{2,k+1},0,\ldots,0\right)^{\top},
\end{equation*} 
 where $\mathbf{H}_{k+1}\mathbf{W}_{k+1}=\mathbf{I}$, and $\mathbf{I} \in\mathbb{R}^{(m+n+1)\times (m+n+1)}$ here is the identity matrix .
 The conclusion is consistent with Theorem \ref{corollary5.4-xpc-2-test}.
\end{remark}

\begin{sloppypar} 
The inverse transformations/encryptions of $\check{T}_k$ and other $T_k$ are difficult to calculate. However, our algorithm DFOp can still solve the problems with such transformations/encryptions exactly. In addition, the result and the discussion above can be applied to  encrypt the functions in the private black-box optimization problems. Moreover, the transformations/encryptions with no change of the model solution shift are exclusive for model-based derivative-free methods, and especially for DFOp.  
\end{sloppypar} 

\subsection{Other notes on transformations/encryptions and DFOp}

We derive the following corollaries to separately analyze the multiplication and the addition in the transformation, which are also prepared for the encryption in the private black-box optimization.
\begin{corollary}\label{xpc-1}
If the objective function {$F$} has been transformed/encrypted to $F_k=C_{1,k}F$ at the $k$-th query/iteration step, where $C_{1,k}\in \mathbb{R}$ is a positive constant, 
%$C_{1,k}>0$, 
for each $ k \in \mathbb{N}^+$, then the corresponding vector $r\in\mathbb{R}^m$ in the model function updating formula of DFOp is
\begin{equation}\label{new-r-1}
\begin{aligned}
r=&\left(
(C_{1,k+1}-C_{1,k})F(x_1),
\ldots,
(C_{1,k+1}-C_{1,k})F(x_{t-1}),
\delta_k,\right.\\
&\left.(C_{1,k+1}-C_{1,k})F(x_{t+1}),
\ldots,
(C_{1,k+1}-C_{1,k})F(x_m)
\right)^{\top},
\end{aligned}
\end{equation}
in which, the parameter $\delta_k$ is defined by (\ref{delta_k}). 
\end{corollary}
\begin{corollary}\label{xpc-2}
If the objective function {$F$} has been transformed/encrypted to be $F_k=F+C_{2,k}$ at the $k$-th query/iteration step, where the constant $C_{2,k}\in \mathbb{R}$, for each $k \in \mathbb{N}^+$, then the corresponding vector $r\in \mathbb{R}^m$ in the model function updating formula of DFOp is
\begin{equation}\label{new-r-2}
r=\left(
C_{2,k+1}-C_{2,k},\ldots,
C_{2,k+1}-C_{2,k},
\delta_k,
C_{2,k+1}-C_{2,k},
\ldots,
C_{2,k+1}-C_{2,k}
\right)^{\top},
\end{equation}
in which, the $t$-th component $\delta_k$ is defined by (\ref{delta_k}) as well.
\end{corollary}

From Corollary \ref{xpc-1} and Corollary \ref{xpc-2}, we can learn that $m-1$ components of the vector $r$ are proportionate in (\ref{new-r-1}) and all the same in (\ref{new-r-2}). The splitted transformations/encryptions are shown in direct forms, which will simplify our computation in the model updating. Results indicate that we can only get the values of $\frac{C_{1,k+1}}{C_{1,k}}$ and $C_{2,k+1}-C_{2,k}$, but can never get the original values of $F(x_i)$ from $F_k(x_i)$, $\forall k\in \mathbb{N}^+$ directly.

 \begin{sloppypar}
Moreover, DFOp is designed not only for problems with linearly transformed/encrypted objective functions, but also for problems of which the objective functions are transformed/encrypted by other transformations. For instance, the  transformations/encryptions can be the random linear transformations/encryptions used in the private black-box optimization, which will be introduced in Section \ref{Solving private black-box optimization problems with DFOp}. In addition, the problems with positive monotonic transformations/encryptions or the transformations/encryptions given in subsection \ref{Transformation maintaining model solution shift} can also be solved by DFOp. The inverse transformations/encryptions of the transformations/encryptions above can not be directly derived. 
Besides, the solver NEWUOA combined with computing the inverse transformations/encryptions at each step will bring a huge amount of computation, especially when there exist random linear transformations/encryptions or other complex transformations. Therefore, DFOp can be widely utilized, since it can solve problems with multiple transformations, without calculating the inverse transformation at each step but only using the uniform simple model updating formula.
\end{sloppypar}
 
\section{More Discussion on Privacy-preserving Black-box Optimization}
\label{Solving private black-box optimization problems with DFOp} 

This section will focus on the details of our new solver DFOp (in the experimental results, it is denoted as NEWUOA-P, reflecting its implementation within the NEWUOA framework for the purpose of direct comparison) when it is applied for solving the private black-box optimization problems. %with our new solver DFOp.
In the step-encryption\footnote{The term step-encryption means that the encryption depends on the iteration/query step.} private problems, the corresponding objective function will transform with the iteration in the output. Thus we try to apply DFOp to solve problem (\ref{1.1}). In order to analyze DFOp's performance when solving the private  black-box optimization problems, we firstly give the following theorem.
\begin{theorem}\label{P_TF}
The encryption of the private black-box optimization problem defined in subsection \ref{Private-BBO} is a transformation.
\end{theorem}

When solving private black-box optimization problems, for the objective function $F=f+h$, the vector $r=r-(0,\ldots,0)^{\top}$ in the model updating formula of DFOp has the form
\begin{equation*}
\begin{aligned}
r=&\left(
h_{k+1}(x_1)-h_k(x_1),
\ldots,
h_{k+1}(x_{t-1})-h_k(x_{t-1}),
\delta_k,\right.\\
&\left. h_{k+1}(x_{t+1})-h_k(x_{t+1}),
\ldots,
h_{k+1}(x_m)-h_k(x_m)
\right)^{\top},
\end{aligned}
\end{equation*}
in which, the parameter $\delta_k$ is defined by
\begin{equation}\label{delta_k-2}
\begin{aligned}
\delta_k= (f(x_{\text{new}})+h_{k+1}(x_{\text{new}})-f(x_{\text{opt}})-h_k(x_{\text{opt}}))
-(Q_{\text{old}}(x_{\text{new}})-Q_{\text{old}}(x_{\text{opt}})).
\end{aligned}
\end{equation}

\begin{remark}\label{callingtime}
Computing the function $h_k$ at $m$ points at the $k$-th step does not cause expensive evaluations for the data provider, since the function evaluation at a point only needs to be done once, and the extra operation is the encryption, which is usually adding a noise. Therefore, providing the function values called by DFOp is acceptable to the data provider. It will be shown that there can be no privacy risk to provide $m$ encrypted function evaluations. In addition, we compare the numbers of calling the function $f$ in the numerical experiments according to the discussion above.
\end{remark}

\subsection{Differentially private noise-adding mechanisms}
\label{Differentially}

The encryption mechanism here should be carefully selected by the data provider.  We propose the encryption mechanisms according to the linear transformation we discussed above and the characteristics of derivative-free optimization. 
 
The output function values called by DFOp, can be regarded as the random output from the similar random distributions, which can prevent attackers. The differentially private mechanisms we are going to propose are linear transformations/encryptions and they hold the optimality-preserving property. 

In order to propose two differentially private mechanisms with the proof of the differential privacy, we present the definition of the neighboring databases.
\begin{definition}\label{adjacent dataset}
Suppose that ${z}$ %=\{\bar{z}_1,\ldots, \bar{z}_n\}$} 
and ${z}'$ 
%=\{\bar{z}'_1,\ldots, \bar{z}'_n\}$} 
are data sets. ${z}$ and ${z}'$ are called neighboring databases if $\vert{z}\triangle {z}'\vert=1$, where $\triangle$ denotes the symmetric difference between two sets, which is defined as ${z}\triangle {z}'=(z\cup z')\setminus(z\cap z')$, and $\vert\cdot\vert$ denotes the cardinality.
\end{definition}

Besides, the definition of differential privacy is given in the following.
\begin{definition} (\cite{Dwork2006})
\label{differential privacy}
A randomized operator $\mathcal{A}$ is ${\varepsilon}_{B}$-differentially private if for all neighboring databases ${z}$, ${z}^{\prime}$ and all sets $\mathcal{S}$ of outputs, it holds that
\begin{equation*}
\operatorname{Pr}(\mathcal{A}({z}) \in \mathcal{S}) \leq e^{{\varepsilon}_{B}} \cdot \operatorname{Pr}(\mathcal{A}({z}^{\prime}) \in \mathcal{S}),
\end{equation*}
where $\operatorname{Pr}(\mathcal{A}({z}) \in \mathcal{S})$ denotes the probability of $\mathcal{A}({z}) \in \mathcal{S}$. The parameter ${\varepsilon}_{B}$ is called the privacy budget, and a smaller privacy budget refers to a safer encryption.
\end{definition}
   
In the following parts, we use two kinds of operators to illustrate the differential privacy of the additive mechanism and the mixed mechanism that will appear then.

\subsection{The differentially private mechanism based on Laplacian noise}

In the private black-box optimization, we assume that the encryption has the form $h_k(x)\ne h(x)$, for arbitrary $x \in \mathbb{R}^n$. We denote the $i$-th interpolation set as $X_i$. The data set is $z_k=\cup_{i=1}^{k}% h(X_i)
\{h(x), \ x \in X_i\}$, 
%where $h(X_i)=\{h(x), \ x \in X_i\}$, 
and we know that for each $k$, $z_{k-1}$ and $z_{k}$ are neighboring databases. Besides, $\hat{x}_j$ denotes the $j$-th point in the ordered set $\cup_{i=1}^{k} X_i$ in this section, and it will be often used in the proof for clarity.

We try to add Laplacian noise by the operator  $\mathcal{A}_k$ to encrypt $h$ at the $k$-th iteration. The operator $\mathcal{A}_k$ is defined in the following.%Definition \ref{definition4.3}. 
\begin{definition}\label{definition4.3}
At the $k$-th iteration, we define the operator $\mathcal{A}_k$ as
$\mathcal{A}_k(z_j)=h(\hat{x}_j)+C\eta_k, \ j=1,\ldots, k$, where $\eta_k \sim \operatorname{Lap}(b_k)$, $b_k>0$, and $C$ is a positive constant. 
\end{definition}

The probability density function of $\operatorname{Lap}(b_k)$ is $p(x)=\frac{1}{2b_k}\cdot e^{-\frac{\vert x \vert}{b_k}}$.
\begin{remark}
Global coefficient $C$ is newly designed  to control the signal-to-noise ratio. 
\end{remark}

To describe the differential privacy of  $\mathcal{A}_k$, we give the following definition.
\begin{definition}
We define the global sensitivity of the function $h$ at the $k$-th iteration as
\begin{equation}\label{4.1}
\text{GS}_{h,k}=\max_{k-m+1\le l \le k-1} \left\vert h(x^*_l)-h(x^*_{l+1}) \right\vert,
\end{equation}
where $x^*_l$ and $x^*_{l+1}$ are exactly corresponding new points in two adjacent ordered iteration points sets $\cup_{i=1}^{l} X_i$ and $\cup_{i=1}^{l+1} X_i$, which satisfy that $\cup_{i=1}^{l+1}X_i=\cup_{i=1}^{l} X_i \cup \{x^*_{l+1}\}$.
\end{definition}

The global sensitivity $\text{GS}_{h,k}$ is used to describe the change of the function values over a set. Then we are going to prove that the operator $\mathcal{A}_k$ is differentially private, which means that it can encrypt the private function $h$. The following lemma follows from basic probability theory. 
\begin{lemma}\label{lemma4.6}
Suppose that $\eta \sim \operatorname{Lap}(b_k)$, $b_k>0$, the constant $M>0$, then for arbitrary constants $a,b \in \mathbb{R}$, which satisfy that $a<b$, the inequality $\operatorname{Pr}(a+M\le\eta \le b+M) \le e^{\frac{M}{b_k}} \cdot \operatorname{Pr}(a\le \eta \le b)$ holds.
\end{lemma}

\begin{proof}
What we need to prove is that, for arbitrary $ s \in \mathbb{R}$, it holds that
\begin{equation*}
p(s) \le e^{\frac{M}{b_k}}\cdot p(s-M),
\end{equation*}
where $p(\cdot)$ is the probability density function of $\operatorname{Lap}(b_k)$. We know that
\begin{equation*}
\frac{p(s)}{p(s-M)}=\frac{\frac{1}{2 b_k} e^{-\frac{\vert s\vert}{b_k}}}{\frac{1}{2 b_k} e^{-\frac{\vert s-M\vert}{b_k}}}.
\end{equation*}
Then, we derive that
\begin{equation*}
 \frac{p(s)}{p(s-M)}=\left\{
\begin{aligned}
&e^{-\frac{M}{b_k}},  \ \ \ \ 0 \leq s-M<s,\\
&e^{\frac{-2 s+M}{b_k}},\  s-M \leq 0 < s\ \text{or}\ s-M < 0 \leq s,\\
&e^{\frac{M}{b_k}},\ \ \ \ \ \ s-M<s \leq 0,
\end{aligned} \right. 
\end{equation*}
which is bounded by
\begin{equation*}
\ \left\{
\begin{aligned}
&e^{-\frac{M}{b_k}},  \ \ \ \ 0 \leq s-M<s,\\
&e^{\frac{M}{b_k}},\ \ \ \ \ \   s-M \leq 0 < s\ \text{or}\ s-M < 0 \leq s,\\
&e^{\frac{M}{b_k}},\ \ \ \ \ \   s-M<s \leq 0.
\end{aligned} \right. 
\end{equation*}
Owing to $e^{\frac{M}{b_k}}>e^{-\frac{M}{b_k}}$, $e^{\frac{M}{b_k}}$ is an upper bound of $\frac{p(s)}{p(s-M)}$. Hence the statement in Lemma \ref{lemma4.6} is proved.
\end{proof}

 With the help of Lemma \ref{lemma4.6}, we can derive that  $\mathcal{A}_k$ is differentially private.
\begin{theorem}\label{4.30-4}
At the $k$-th iteration with $m$ interpolation points $\hat{x}_{k-m+1},\ldots,\hat{x}_k$, if $\eta_k$ only depends on the number of iterations $k$, then the operator $\mathcal{A}_k$ is $\varepsilon_k$-differentially private, where $\varepsilon_k=\frac{\text{GS}_{h,k}}{b_kC},$ and $\text{GS}_{h,k}$ is defined as (\ref{4.1}).
\end{theorem}
\begin{proof}
According to the definition of $\mathcal{A}_k$, 
it is $\varepsilon_k$-differentially private if 
\begin{equation*}
\operatorname{Pr}(\tilde{a}\le h(\hat{x}_j)+C\eta_k \le \tilde{b})
 \le e^{\varepsilon_k} \cdot 
\operatorname{Pr}(\tilde{a}\le h(\hat{x}_{j+1})+C\eta_k \le \tilde{b}),
\end{equation*} 
for arbitrary $\tilde{a}\le \tilde{b}$. Therefore, what we need to prove is that the inequality
\begin{equation}\label{lemma1-1}
\operatorname{Pr}(a+\frac{\text{GS}_{h,k}}{C} \le\eta_k\le b+\frac{\text{GS}_{h,k}}{C}) \le e^{\frac{\text{GS}_{h,k}}{b_k C}} \cdot \operatorname{Pr}(a \le\eta_k\le b)
\end{equation}
holds, and for $M^\prime<\frac{\text{GS}_{h,k}}{C}$, the inequality
\begin{equation}\label{lemma1-2}
\operatorname{Pr}(a+M^\prime \le\eta_k\le b+M^\prime) \le e^{\frac{\text{GS}_{h,k}}{b_k C}} \cdot \operatorname{Pr}(a \le\eta_k\le b)
\end{equation}
also holds, where $a\le b$. By the virtue of Lemma \ref{lemma4.6}, we know that inequality (\ref{lemma1-1}) holds. According to the properties of Laplacian distribution, inequality (\ref{lemma1-2}) also holds when $M^\prime<\frac{\text{GS}_{h,k}}{C}$. The proof is completed.
\end{proof}

We have proved that the encryption $h_k(x)=h(x)+C\eta_k, \ \eta_k \sim \operatorname{Lap}(b_k)$, $C > 0$, is a kind of differentially private mechanism, and we call this the additive mechanism.

\subsection{The differentially private mechanism based on mixed noise}

We have discussed how to add random additive noise to encrypt $h$. We now try to provide another choice by adding random multiplicative noise in the following. 
We now give the definition of the operator $\mathcal{B}_k$. 
\begin{definition}
At the $k$-th iteration, we define the operator $\mathcal{B}_k$ as
$
\mathcal{B}_k (z_j)=h(\hat{x}_j)+\gamma_k \cdot (f(\hat{x}_j)+h(\hat{x}_j)),
$ 
where $\gamma_k \sim \operatorname{U}(-u_k, u_k)$, $0<u_k <1$, $j=1,\ldots, k$,  where $\operatorname{U}$ denotes the uniform distribution, and the corresponding probability density function is
\begin{equation*}
p(x)=\left\{
\begin{aligned}
&\frac{1}{2u_k},\ \ \text{if $x\in [-u_k, u_k]$,}\\
&\ \ \ 0,\ \ \ \ \text{otherwise}.
\end{aligned}
\right.
\end{equation*}
 We assume that $\vert f(\hat{x}_j)+h(\hat{x}_j)\vert>0$, for each $j\in\mathbb{N}^{+}$, since $f(\hat{x}_j)+h(\hat{x}_j)$ has no sense of the encryption otherwise.
\end{definition}
\begin{remark}
We focus on the common definitional domain of the discussed uniform distributions.
\end{remark}

Then we can prove that the operator $\mathcal{B}_k$ is differentially private.
\begin{theorem}\label{multi1}
At the $k$-th iteration, if $\gamma_k$ only depends on the iteration $k$, then the operator $\mathcal{B}_k$ is $\varepsilon_k'$-differentially private, where 
$
\varepsilon_k'=\max_{k-m+1 \le j \le k-1} \vert \ln \frac{\vert f(\hat{x}_{j+1})+h(\hat{x}_{j+1})\vert}{\vert f(\hat{x}_{j})+h(\hat{x}_{j})\vert} \vert$. 
\end{theorem}
\begin{proof}
At the $k$-th iteration with $m$ interpolation points $\hat{x}_{k-m+1},\ldots,\hat{x}_k$, it holds that   
\begin{equation*}
\frac{p(h(\hat{x}_j)+\gamma_k \cdot (f(\hat{x}_j)+h(\hat{x}_j)))}{p(h(\hat{x}_{j+1})+\gamma_k \cdot (f(\hat{x}_{j+1})+h(\hat{x}_{j+1})))}=\frac{\vert f(\hat{x}_{j+1})+h(\hat{x}_{j+1})\vert}{\vert f(\hat{x}_{j})+h(\hat{x}_{j})\vert},
\end{equation*}
where $p(\cdot)$ denotes the probability density function of $\gamma_k$. 
Thus
\begin{equation}\label{601}
\begin{aligned}
\frac{p(h(\hat{x}_j)+\gamma_k \cdot(f(\hat{x}_j)+h(\hat{x}_j)))}{p(h(\hat{x}_{j+1})+\gamma_k \cdot(f(\hat{x}_{j+1})+h(\hat{x}_{j+1})))} \le
e^{\underset{k-m+1 \le j \le k-1}{\max} \left\vert \ln \frac{\left\vert f(\hat{x}_{j+1})+h(\hat{x}_{j+1})\right\vert}{\left\vert f(\hat{x}_{j})+h(\hat{x}_{j})\right\vert} \right\vert}
=e^{\varepsilon'_k}.
\end{aligned}
\end{equation}
Therefore,  we obtain that
\begin{equation*}
{p(h(\hat{x}_j)+\gamma_k \cdot(f(\hat{x}_j)+h(\hat{x}_j)))} \le e^{\varepsilon'_k}\cdot {p(h(\hat{x}_{j+1})+\gamma_k \cdot(f(\hat{x}_{j+1})+h(\hat{x}_{j+1})))}.
\end{equation*}
The proof is completed.%}
\end{proof}

We finally remark that the mix of the operators $\mathcal{A}_k$ and $\mathcal{B}_k$ is still differentially private, which demonstrates that the mixed mechanism proposed by us can be used to encrypt $h$ in practice. %Theorem \ref{theorem4.18} 
The following theorem follows from Theorem \ref{4.30-4}, Theorem \ref{multi1} and Claim 2.2 in the paper of Kasiviswanathan et al. \cite{Kasiviswanathan2008}.
\begin{theorem}\label{theorem4.18}
Let the mixed operator $\mathcal{M}_k$ be obtained by mixing operators $\mathcal{A}_k$ and $\mathcal{B}_k$, i.e.,  $\mathcal{M}_k=p(\mathcal{A}_k, \mathcal{B}_k)=\mathcal{A}_k$ or $\mathcal{B}_k$ randomly, $k \in \mathbb{N}^+$. Then the operator $\mathcal{M}_k$ is $\bar{\varepsilon}_k$-differentially private, where
\begin{equation}\label{4.13}
\begin{aligned}
\bar{\varepsilon}_k=
\frac{\text{GS}_{h,k}}{b_kC}+
\max_{k-m+1 \le j \le k-1} \left\vert \ln \frac{\vert f(
\hat{x}_{j+1})+h(\hat{x}_{j+1})\vert}{\vert f(\hat{x}_{j})+h(\hat{x}_{j})\vert} \right\vert. 
\end{aligned}
\end{equation}
\end{theorem}

\begin{proof}
This is a consequence of Theorem \ref{4.30-4} and Theorem \ref{multi1}, according to Claim 2.2 about the mixed operator in the paper of Kasiviswanathan et al. \cite{Kasiviswanathan2008}.
\end{proof}

This mechanism is called the mixed mechanism. Theorem \ref{theorem4.18} indicates that the data provider can carefully select the operator $\mathcal{M}_k$, the parameters $b_k$ and $C$ to keep the privacy budget $\bar{\varepsilon}_k$ small. In detail, according to (\ref{4.13}), the data provider can define the constant $C$ large enough to make
$ 
%\max_{j \le k\le j+m-1} 
\vert\frac{\text{GS}_{h,k}}{b_kC}\vert
$
as small as possible.  
The data provider is eligible to instruct the random operator $\mathcal{M}_k$ not to choose the operator $\mathcal{B}_k$, or to choose $\mathcal{B}_k$ under lower probability when the value of the privacy budget $\bar{\varepsilon}_k$ is larger than an expected constant, and thus the value of the last term of (\ref{4.13}) can be controlled. As a result, the scale of the privacy budget can be small, achieving a pretty safe encryption.

In addition, notice that for the mixed mechanism, the noise at each point is $C \eta_k$ or $\gamma_k \cdot (f(x)+h(x))$, where $\eta_k\sim\operatorname{Lap}(b_k), C>0$, and $\gamma_k\sim \operatorname{U}(-u_k,u_k)$, where $0<u_k<1$. The specific form is hard to obtain by subtracting or dividing directly at an iteration. This characteristic makes the noise of every interpolation point different from each other.   

\subsection{Convergence analysis for solving differentially private problems}

This subsection analyzes the global convergence of DFOp when it is applied for solving differentially private problems. The following theorem shows that the differentially private problems belong to the problems that the objective functions are with positive monotonic transformations. It should be noted that the convergence analysis in this work is carried out under stronger theoretical assumptions and does not directly apply to the practical NEWUOA framework.

\begin{theorem}\label{EM_LT}
The function values provided by the additive mechanism and the mixed mechanism are both transformed/encrypted from the original function values by positive monotonic transformations.
\end{theorem}

\begin{proof}
We only need to prove that the function values provided by $\mathcal{A}_k$ and $\mathcal{B}_k$ are transformed/encrypted by positive monotonic transformations. 
It holds for each $j\in \mathbb{N}^{+}$ that 
\begin{equation*}
\begin{aligned}
&\mathcal{A}_k(z_j)=h(\hat{x}_j)+C\eta_k,\\
 &\mathcal{B}_k(z_j)=h(\hat{x}_j)+\gamma_k(f(\hat{x}_j)+h(\hat{x}_j)), 
\end{aligned}
\end{equation*} 
where $\eta_k \sim \operatorname{Lap}(b_k)$, $b_k > 0$, $C > 0$ and $\gamma_k \sim \operatorname{U}(-u_k, u_k)$, $0<u_k <1$.

Therefore, for $h(\hat{x}_p)<h(\hat{x}_q)$, it holds that
\begin{equation*}
\begin{aligned}
h(\hat{x}_p)+C\eta_k&<h(\hat{x}_q)+C\eta_k,\\
h(\hat{x}_p)+\gamma_k(f(\hat{x}_p)+h(\hat{x}_p))&<h(\hat{x}_q)+\gamma_k(f(\hat{x}_q)+h(\hat{x}_q)). 
\end{aligned}
\end{equation*} 

Consequently Theorem \ref{EM_LT} is true. 
\end{proof}

Next, we obtain the following lemma naturally.
\begin{lemma}\label{corollary5.5}
The additive mechanism and the mixed mechanism do not make the solution shift change.
\end{lemma}

\begin{assumption}\label{10.3-new}
Suppose that \(x^*\) and \(\Delta\) are given. Assume that \(F_k\) are continuously differentiable with Lipschitz continuous gradient in an open domain containing the set \(L_{\text{enl},k}(x^*) ,\)  in which 
$
L_{\text{enl},k}(x^*)
=\cup_{x \in L_{k}(x^*)} \mathbb{B}_{\Delta}(x)$, and $ L_{k}(x^*)=\{x \in \mathbb{R}^{n}: F_k(x) \leq F_k(x^*)\}$, $\forall k\in\mathbb{N}^+
$.
\end{assumption}

In addition, considering the minimization, we assume that each transformed function $F_k$ is bounded from below. 
\begin{assumption}\label{10.5-new}
Suppose that \(F_k\) are all bounded from below on \(L_{k}(x^*)\), that is, there exists a
constant \(\kappa_{*}\) such that, for all \(x \in L_{k}(x^*)\), \(F_k(x) \geq \kappa_{*}\), $\forall k\in \mathbb{N}^+$.
\end{assumption}

For simplicity, we let the Hessian of the model function, denoted as $\nabla^{2} Q_{k}$, be uniformly bounded. The following assumption shows the details.
\begin{assumption}\label{10.6-new}
There exists a constant \(\kappa_{\text{bhm}}>0\) such that, for all \(x_{k}\) generated by DFOp, 
$
\Vert \nabla^2 Q_{k}(x_k)\Vert_{2} \leq \kappa_{\text{bhm}}.
$
\end{assumption}

\begin{proof}
The conclusion is derived from Theorem \ref{monotone} and Theorem \ref{EM_LT} directly.
\end{proof}

\begin{theorem}\label{con-new}
Let Assumption \ref{10.3-new}, Assumption \ref{10.5-new} and Assumption \ref{10.6-new} hold. Then
$
\lim_{k \rightarrow\infty}$ $\nabla F_k(x_{k})=0
$
holds for DFOp, where $F_k(x)=T_k(F(x))$, and $T_k$ is a positive monotonic transformation for each $k\in\mathbb{N}^+$. 
\end{theorem}
\begin{proof}
The proof process is as same as the convergence  analysis for trust-region methods based on derivative-free models in Section 10.4 in the book of Conn, Scheinberg and Vicente  \cite{conn2009introduction}. 
\end{proof}

\begin{corollary}\label{corollary-con}
Let Assumption \ref{10.3-new}, Assumption \ref{10.5-new} and Assumption \ref{10.6-new} hold, and we assume that there exists $\varepsilon>0$, such that ${\lim\inf}_{k\rightarrow\infty}\frac{dF_k}{dF}>\varepsilon$, then
$
\lim_{k \rightarrow \infty} \nabla F(x_{k})=0.
$
\end{corollary}

\begin{proof}
Theorem \ref{con-new} shows that
\begin{equation*}
\lim_{k \rightarrow \infty} \nabla F_{k}(x_{k})=0.
\end{equation*}
Besides, we know that
\begin{equation*}
\begin{aligned}
\nabla F_{k}(x_{k})=\frac{d F_{k}}{d F} \cdot \nabla F(x_{k}),
 \end{aligned}
 \end{equation*}
 \begin{equation*}
\begin{aligned}
 \underset{k\rightarrow\infty}{\lim\inf}\frac{dF_k}{dF}>\varepsilon,
 \end{aligned}
 \end{equation*}
  and thus 
\begin{equation*}
\lim_{k \rightarrow \infty} \nabla F(x_{k})=0.
\end{equation*} 
The proof is completed.
\end{proof}

Therefore, the following theorem holds, as a special case of Corollary \ref{corollary-con}.
\begin{theorem}\label{xxxpppccc-newthm}\color{black}
Let Assumption \ref{10.3-new}, Assumption \ref{10.5-new} and Assumption \ref{10.6-new} hold. Then it holds that 
 $
\lim_{k \rightarrow\infty} \nabla F(x_{k})=0
 $ 
 when using DFOp, where the transformed/encrypted objective functions satisfy that $F_k=f+h_k$, $\forall k\in \mathbb{N}^+$, where $h_k$ is encrypted by the additive mechanism or the mixed mechanism.
\end{theorem}

\begin{proof}
{\color{black}Corollary \ref{corollary-con},} Theorem \ref{EM_LT} and Lemma \ref{corollary5.5} give the conclusion of Theorem \ref{xxxpppccc-newthm}.
\end{proof}

According to the global convergence analysis of our DFOp, it can achieve good convergence property as an optimization solver for private black-box problems with additive or mixed noise-adding mechanisms.%, while other algorithms can not. 

\subsection{Numerical results}\label{Numerical}

This part compares the results in numerical experiments of DFOp with those of other derivative-free solvers when solving private black-box optimization problems encrypted by differentially private mechanisms. The results can numerically demonstrate the main features and advantages of DFOp.

We compare DFOp with the following algorithms: MATLAB optimization package Fminunc \cite{Fminunc-1,Fminunc-2}, MATLAB optimization package Fminsearch \cite{Fminsearch}, Nelder-Mead simplex method \cite{Kelley1999a,Higham2000} and NEWUOA. We use the performance profile \cite{audet2017derivative} for the comparison of different algorithms. 
The performance profile describes the number of iterations taken by the algorithm in the algorithm set $\mathcal{A}$ to achieve a given accuracy when solving problems in a given problem set. 
We define the value $F_{\mathrm{acc}}^{N}=(F(x_{N})-F(x_{0}))/(F(x_{\text{best}})-F(x_{0})) \in [0,1]$, and the tolerance $\tau \in [0,1]$, where $x_{N}$ denotes the best point found by the algorithm after $N$ function evaluations, $x_{0}$ denotes the initial point, and $x_{\text{best}}$ denotes the best known solution. When $F_{\mathrm{acc}}^{N} \ge 1-\tau$,  we say that the solution reaches the accuracy $\tau$. We give $N_{s,p}=\min\{n \in \mathbb{N},\ F_{\mathrm{acc}}^{n}\ge 1-\tau \}$ and the following definitions.
\begin{equation*}
\begin{aligned}
T_{s, p}=\left\{\begin{array}{ll}
1, & \text{if} \ F_{\mathrm{acc}}^{N} \geq 1-\tau \ \text{for some } N,\\
0, & \text{otherwise},
\end{array}\right.
\end{aligned}
\end{equation*}
and
\begin{equation*}
\begin{aligned}
r_{s, p}=\left\{\begin{array}{ll}
\frac{N_{s, p}}{\min \left\{N_{\tilde{s}, p}:\ \tilde{s} \in \mathcal{A}, T_{\tilde{s}, p}=1\right\}},& \text {if} \ T_{s, p}=1, \\
\ \ \ \ \ \ \ \ \ \ +\infty, & \text {if} \ T_{s, p}=0,
\end{array}\right.
\end{aligned}
\end{equation*}
where $s$ is the given solver or algorithm. For the given tolerance $\tau$ and a certain problem $p$ in the problem set $\mathcal{P}$, the parameter $r_{s, p}$ shows the ratio of the number of the function evaluations using the solver $s$ divided by that using the fastest algorithm on the problem $p$. In the performance profile, $\pi_{s}(\alpha)=\frac{1}{\vert\mathcal{P}\vert}\left\vert\left\{p \in \mathcal{P}: r_{s, p} \leq \alpha\right\}\right\vert$, where  $\alpha \in [1, +\infty)$, and $\vert\cdot\vert$ denotes the cardinality. Notice that a higher value of $\pi_s(\alpha)$ represents solving more problems successfully. 

We apply five algorithms to solve the private black-box optimization problems. All problems have been encrypted by the additive mechanism or the mixed mechanism. Besides, we also plot NEWUOA-N in the figure, and the problems solved by NEWUOA-N have no noise, i.e., they are not being encrypted. {\color{black}The notation -N here denotes no noise, and NEWUOA-N} can be regarded as the ground truth in some senses, which acts as a standard answer when the noise does not exist. Furthermore, the comparison between DFOp and NEWUOA-N can show whether DFOp is able to reduce the influence of the encryption noise. The private black-box function parts of the 94 test problems  are shown in Table \ref{table5}, and they are from classical and common unconstrained optimization test functions collections \cite{Powell04onupdating,Powell2003,Conn1994,Toint1978,Li2009,Luksan2010,Andrei2008,Li1988,CUTEr}.

\begin{table}[h!]  
	\centering   
		\caption{Test problems for private black-box function $h$\label{table5}} 
	   \fontsize{8}{8}\selectfont
	\begin{tabular}{lllllll}  
		\hline  
     arglina&  arglina4&  arglinb &  arglinc & argtrig& arwhead&  bdqrtic\\
		bdqrticp &bdalue& brownal & broydn3d &  broydn7d & brybnd&  chainwoo\\
     chebquad &chnrosnbz & chpowellb &chpowells &chrosen &cosinecube& curly10 \\
     curly20 &curly30 &dixmaane & dixmaanf &dixmaang &dixmaanh &dixmaani \\
     dixmaanj &dixmaank &dixmaanl &dixmaanm & dixmaann &dixmaano &dixmaanp \\
     dqrtic &edensch &eg2 &engval1 &errinros & expsum &extrosnb \\
      exttet &firose &fletcbv2 &fletcbv3 &fletchcr &fminsrf2 & freuroth\\
       genbrown &genhumps &genrose &indef &integreq &liarwhd &lilifun3 \\
		  lilifun4 &morebv &morebvl &ncb20 &ncb20b &noncvxu2 &noncvxun \\
       nondia & nondquar &penalty1 &penalty2 &penalty3 &penalty3p &powellsg \\
       power &rosenbrock  &sbrybnd &sbrybndl &schmvett &scosine &scosinel \\
       serose &sinquad &sparsine & sparsqur &sphrpts &spmsrtls &srosenbr \\
       stmod &tointgss &tointtrig &tquartic& trigsabs &trigssqs &trirose1 \\
       trirose2 &vardim &woods & - & -& - & -  \\ 
		\hline                                      
	\end{tabular}
\end{table}

Indeed, numerical results will not be influenced by the specific form of the public black-box function $f(x)$, and the primary influence is caused by the difference between $h_k(x)$ and $h(x)$, which is the noise for the  encryption. As we have mentioned in Remark \ref{callingtime}, the numbers of calling the function $f$ can represent the number of the whole function evaluations.
\begin{remark}
We implement DFOp, NEWUOA and NEWUOA-N in Fortran, while numerical experiments of other algorithms are implemented in MATLAB. However, the difference in program languages makes no influence or unfairness on comparing the number of function evaluations. 
\end{remark}
 
The increase of the scale of the encryption noise $C$ has very little influence on DFOp. However, for other algorithms, the performances have a significant decline when $C$ grows from 1 to 100. 
At the same time, most other algorithms perform worse for solving problems with the mixed mechanisms than solving those with the additive mechanisms. 
Nevertheless, the performance of DFOp is almost the same,  
%in the two figures, 
which means that DFOp is more stable for solving problems with different mechanisms than the other algorithms. 
 
\begin{figure}[htbp]
  \centering

  \begin{subfigure}[t]{0.48\linewidth}
    \centering
    \includegraphics[width=\linewidth]{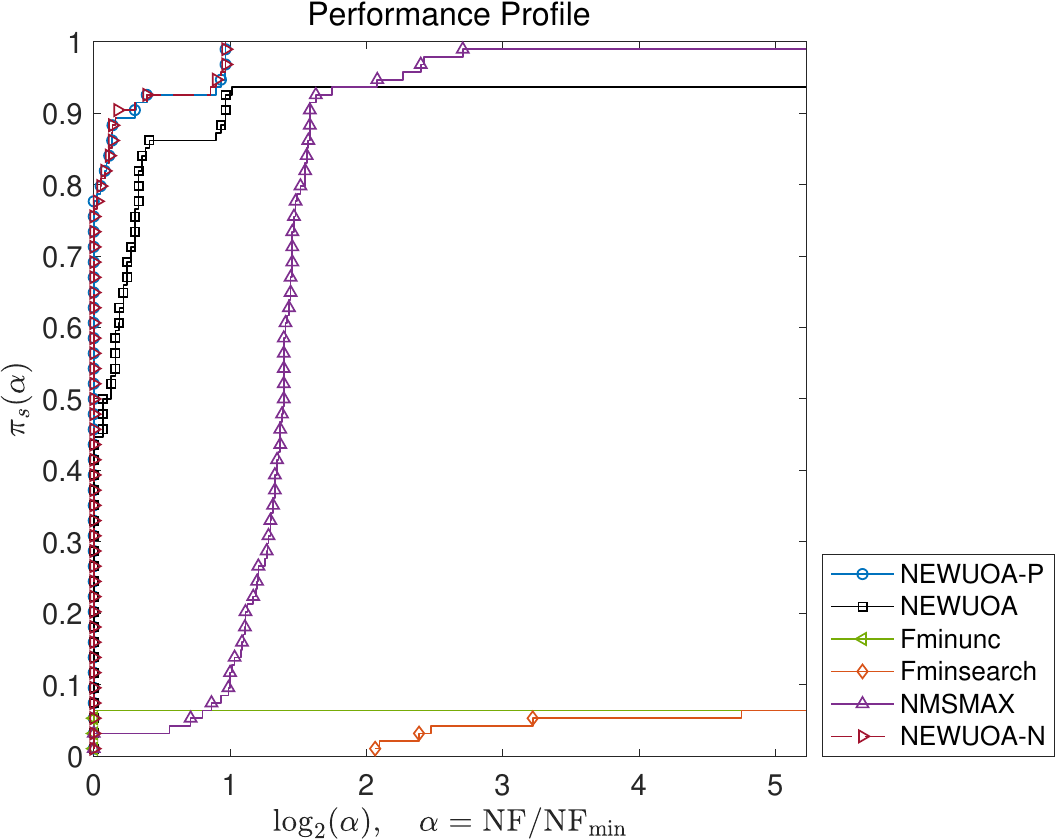}
    \caption{Mixed mechanisms, $\eta_k\sim\operatorname{Lap}\!\bigl(\tfrac{100}{k}\bigr)$,
    $\gamma_k\sim \mathcal{U}\!\bigl(-\tfrac{1}{k},\tfrac{1}{k}\bigr)$,
    $C=100$, $\tau=10^{-4}$.}
    \label{e}
  \end{subfigure}\hfill
  \begin{subfigure}[t]{0.48\linewidth}
    \centering
    \includegraphics[width=\linewidth]{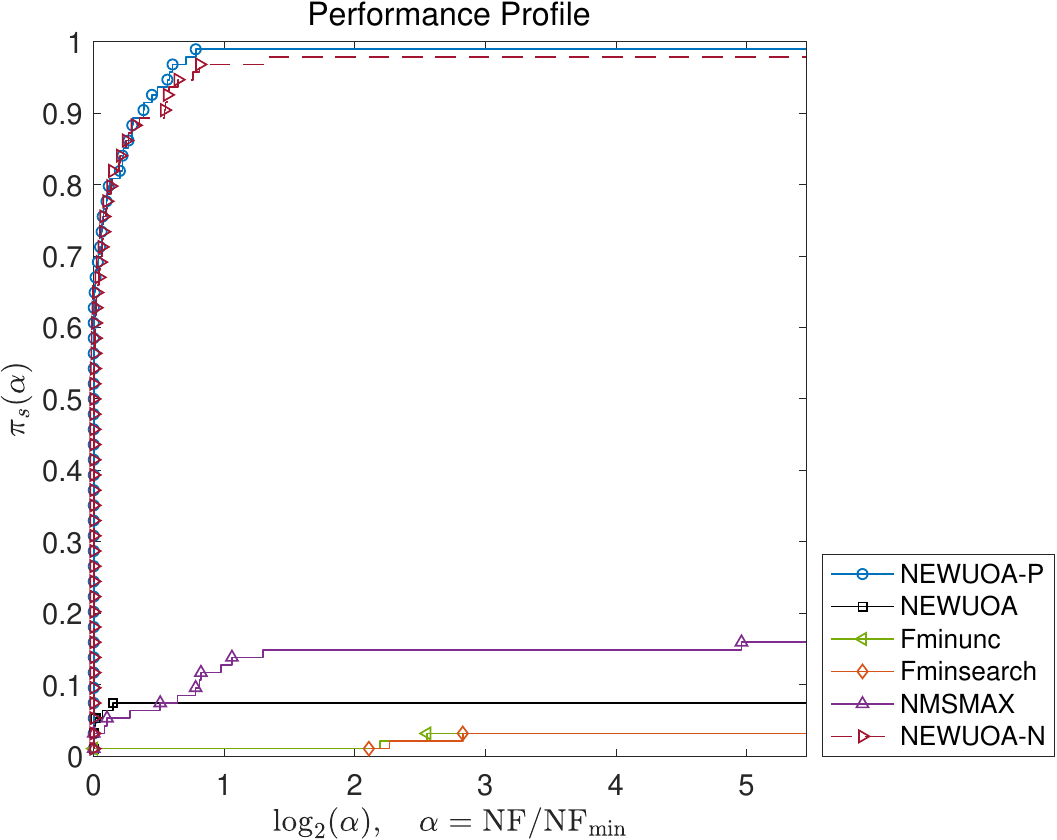}
    \caption{Mixed mechanisms, $\eta_k\sim\operatorname{Lap}\!\bigl(\tfrac{100}{k}\bigr)$,
    $\gamma_k\sim \mathcal{U}\!\bigl(-\tfrac{1}{k},\tfrac{1}{k}\bigr)$,
    $C=100$, $\tau=10^{-5}$.}
    \label{f}
  \end{subfigure}

  \caption{The comparison of different solvers solving the 94 test problems.}
  \label{fig:7}
\end{figure}

In Figs. \ref{e} and \ref{f}, we apply the mixed mechanisms, where $\eta_k\sim\operatorname{Lap}(\frac{100}{k})$, $\gamma_k\sim\operatorname{U}(-\frac{1}{k}, \frac{1}{k})$, $C=100$, and the tolerance $\tau$ are $10^{-1}$ and $10^{-5}$ respectively. It can be observed that when the tolerance $\tau=10^{-1}$, some solvers perform not so  badly, and the difference between the performances of different algorithms is not very large. However, when the tolerance $\tau=10^{-5}$, only DFOp  performs well for encrypted problems. In contrast, the other solvers have a significant decline in the performance profiles. The comparison between Figs. \ref{e} and \ref{f} reveals that even if the requirement of the accuracy is high, DFOp can solve most of the private black-box optimization problems. When the tolerance $\tau$ is small, DFOp performs better than the other algorithms. We can see that the interference of the encryption noise barely affects DFOp since the curves of DFOp and NEWUOA-N are close to each other.  
  
\begin{table}[htbp]
    \fontsize{6}{6}\selectfont
\caption{Parameters and results\label{table6}}
\begin{center}
\begin{tabular}{p{15.4mm}|p{3.2mm}p{13.5mm}p{1.5mm}|p{3.2mm}p{13.5mm}p{0.5mm}|p{3.2mm}p{13.5mm}p{1.5mm}}
\hline
&\multicolumn{3}{c}{$\eta_k \sim \operatorname{Lap}(\frac{1}{k})$}
&\multicolumn{3}{c}{$\eta_k \sim \operatorname{Lap}(\frac{100}{k})$}
&\multicolumn{3}{c}{$\eta_k \sim \operatorname{Lap}(\frac{10}{k})$}\\
\hline
Algorithm  & NF & $F_{\text{opt}}$ & & NF & $F_{\text{opt}}$ & & NF & $F_{\text{opt}}$ & \\ 
\hline
  DFOp  
  & 1033 & 1.25066$\times 10^{-13}$ &  Y
&1046 & 5.43843$\times 10^{-13}$ &  Y
&847 & 6.17840$\times 10^{-13}$ &  Y\\
  NEWUOA-N   
 & 990 & 0 & Y
  &   990 & 0 & Y
& 990 & 0 & Y\\
 NEWUOA  
& 613 & 0.13747 &  N
 & 348 & 7.2318 &  N
& 542 & 1.5818 &  N\\
Fminunc  
& 213 & 1.01$\times 10^5$ & N
& 198 & 1.01$\times 10^5$ & N
& 209 & 1.01$\times 10^5$ & N\\
  Fminsearch  
& 40010 & 0.96864 & N
  & 40003 & 6.1153 & N
& 40004 & 0.31475 & N\\
NMSMAX  
 & 643 & 0.00824 & N
&    612 & 0.39763 & N
& 631 & 0.15528 & N\\
\hline
\hline
&\multicolumn{3}{c}{$\gamma_k \sim \operatorname{U}(-\frac{1}{k},\frac{1}{k})$}
&\multicolumn{3}{c}{$\eta_k \sim \operatorname{Lap}(\frac{100}{k})$, $\gamma_k \sim \operatorname{U}(-\frac{1}{k},\frac{1}{k})$}
&\multicolumn{3}{c}{$\eta_k \sim \operatorname{Lap}(\frac{100}{k})$, $\gamma_k \sim \operatorname{U}(-\frac{k}{10^4},\frac{k}{10^4})$}\\
\hline
Algorithm  & NF & $F_{\text{opt}}$ & & NF & $F_{\text{opt}}$ & & NF & $F_{\text{opt}}$ & \\ 
\hline
  DFOp   
  & 1055 & 2.00320$\times 10^{-13}$ &  Y
  & 1056 & 1.66350$\times 10^{-13}$ &  Y
  & 948 & 5.37654$\times 10^{-13}$ &  Y \\
  NEWUOA-N   
  &  990 & 0 & Y
  & 990 & 0 & Y
  &  990 & 0 & Y\\
 NEWUOA  
 &   408 & 0.73448 &  N
& 432 & 6.533 &  N
 &  409 & 4.0762 &  N\\
Fminunc  
& 187 & 1.01$\times 10^5$ & N
&  143 & 1.01$\times 10^5$ & N
& 33 & 2.22$\times 10^{-16}$ & Y\\
  Fminsearch   
  &  40007 & 1.01$\times 10^5$ & N
& 40007 & 1.02$\times 10^5$  & N
  &  40004 & 6.3653 & N\\
NMSMAX  
& 564 & 0.04238 & N 
 & 585 & 0.06977 & N 
&  609 & 0.02855 & N \\
\hline
\end{tabular}
\end{center}
\end{table}

In the numerical experiments corresponding to Table \ref{table6}, the public black-box function $f$ is $\sum_{i=1}^{10} x_i^4$, and the private black-box function $h$ is $\sum_{i=1}^{10} x_i^2$. Notice that we denote $x=(x_1,\ldots,x_n)^{\top}$ in the part of numerical experiments. Besides, the initial point is $(10,\ldots,10)^{\top}$, and the constant $C=1$. The analytic solution of the numerical experiments is $(0,\ldots,0)^{\top}$, of which the corresponding minimum function value $F_{\text{opt}}$ is $0$. The notations Y and N in Table \ref{table6} denote whether the algorithm solves the problem successfully. Y means that the value of $F_{\text{opt}}$ is less than $10^{-3}$, and N denotes a failure to reach that accuracy. Here the notation NF denotes the number of the function evaluations until the iteration terminates.

From Table \ref{table6}, our numerical results indicate that except DFOp, all other algorithms can hardly solve these relatively basic and easy private problems. In other words, all the algorithms we tested, except DFOp and NEWUOA-N, have poor performances when solving private black-box optimization problems. In addition, the results of NEWUOA-N and DFOp are similar, which show that DFOp solves the optimization problems with encryption quite well, since NEWUOA-N plays a role as the ground truth.

\begin{figure}[h!]
\centering
\includegraphics[width=0.73\textwidth]{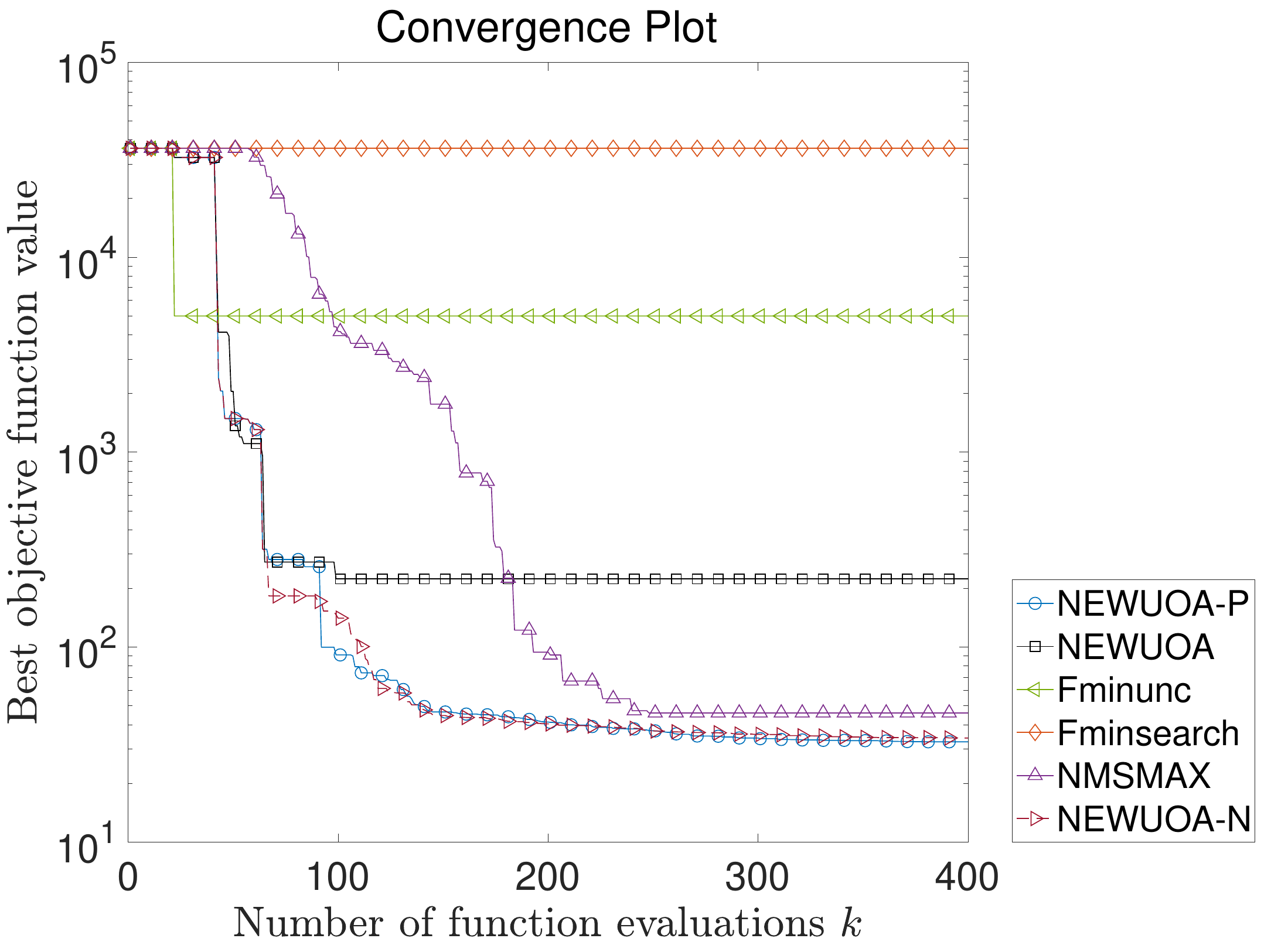}
\caption{Numerical results of algorithms for the 20-dim test problem encrypted arglina, mixed mechanisms: $\eta_k \sim \operatorname{Lap}(\frac{100}{k}), \gamma_k \sim \operatorname{U}(-\frac{1}{k}, \frac{1}{k}),$ $C=100$, the objective function: $f+h$  (note the vertical axis is on a log-scale)\label{fig8}}
\end{figure}

In Fig. \ref{fig8}, $h$ is the test function arglina in CUTEr \cite{CUTEr}, which is
$
h(x)=\sum_{i=1}^{20}(x_i-\bar{x}-1)^2+20(\bar{x}+1)^2
$, 
where $\bar{x}=\frac{1}{20}\sum_{i=1}^{20}x_i$, and $f(x)= 100(\sum_{i=1}^{20} x_i^2 -1)^2$. Besides, the initial point is $(1,\ldots,1)^{\top}$. It can be observed that the algorithms Fminunc and Fminsearch perform quite poorly.  
There is a large gap between the reached accuracy of other solvers and that of DFOp. We can also see that the curves of DFOp and NEWUOA-N decline faster than the other algorithms and reach good solutions. In this private optimization problem, DFOp behaves like the ground truth NEWUOA-N, which means that DFOp has overcome the influence brought by the encryption noise.

\begin{figure}[htbp]
\centering
\includegraphics[width=0.73\linewidth]{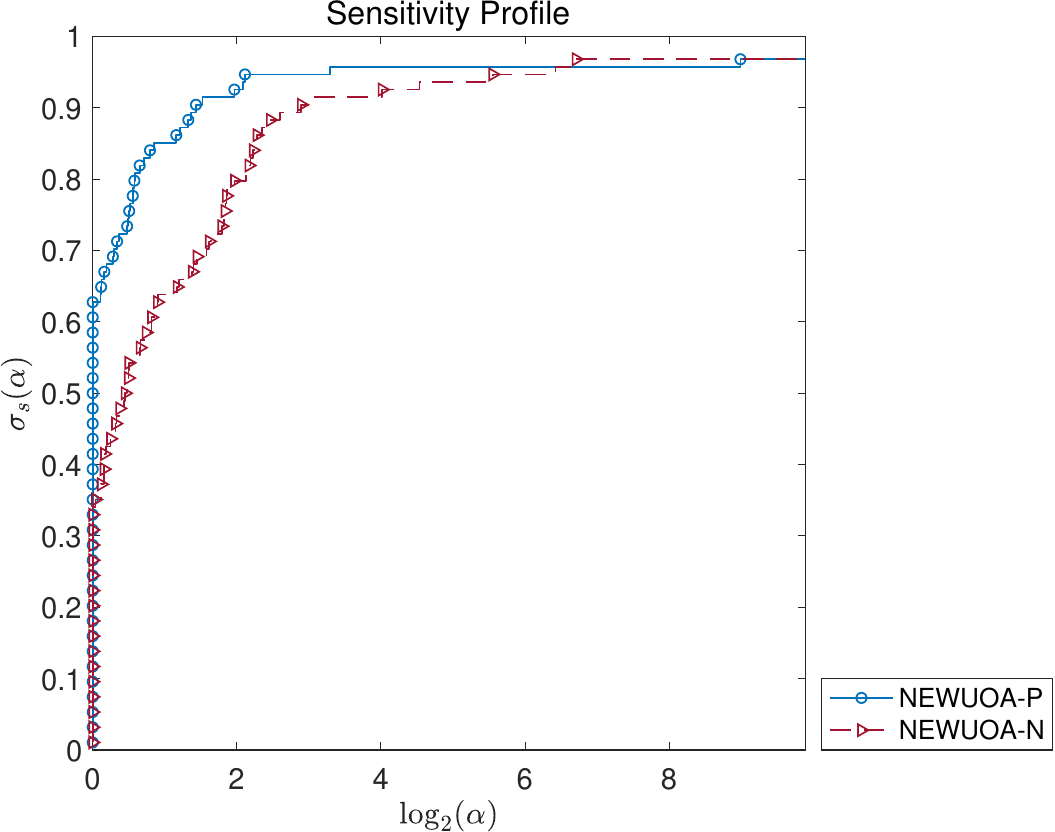}
\caption{Sensitivity profile: mixed mechanisms, $\eta_k \sim \operatorname{Lap}(\frac{100}{k})$, $\gamma_k \sim \operatorname{U}(-\frac{1}{k}, \frac{1}{k})$, $C=1$, $\tau=10^{-5}$\label{fig9}}
\end{figure}

In the numerical results reported in Fig. \ref{fig9}, $h$ denotes 94 test problems in Table \ref{table5} separately, whose dimension is 10, and $f(x)=100(\sum_{i=1}^{10} x_i^2 -1)^2$. In addition, the initial point is $(1,\ldots,1)^{\top}$, and $C=1$. Since a higher value of $\sigma_s(\alpha)$ refers to a more stable algorithm in sensitivity profile \cite{012}, Fig. \ref{fig9} depicts that the performance of DFOp is better than that of NEWUOA-N, which means that the round-off error of DFOp is smaller than that of our compared ground truth NEWUOA-N. The reason perhaps is that the private black-box optimization problems already have random noise terms. In fact, sensitivity profile is another important criterion which evaluates the stability of the algorithms. We denote random permutation matrices as $P_i\in\mathbb{R}^{n\times n}$, $i=1,2, \ldots, M$. In the experiments, $n=10$, $M=100$, and an example of the random permutation matrices is $P_1=\left(e_1, e_2,e_4,e_3,e_8,e_5,e_9,e_{10},e_6,e_7\right)^{\top}$. 
Besides, we define
$
\text{NF}=\left(\text{NF}_{1}, \cdots, \text{NF}_{M}\right),
$ 
where $\text{NF}_{i}$ denotes the number of function evaluations when solving the corresponding problem
$
\min_{x \in \mathbb{R}^{n}} F(P_ix)
$. 
We define 
$  
\operatorname{mean}(\text{NF})$ as $\frac{1}{M} \sum_{i}^{M} \text{NF}_{i}$, and the standard deviation $\operatorname{std}(\text{NF})$ as $\sqrt{\frac{1}{M} \sum_{i}^{M}(\text{NF}_{i}-\operatorname{mean}(\text{NF}))^{2}}$.  
In the performance profile, we apply $\operatorname{std}(\text{NF})$, corresponding to solving the problem $p$ with the solver $s$, instead of $N_{s,p}$ to get $\sigma_s(\alpha)$ instead of $\pi_s(\alpha)$, and then we get the sensitivity profile. 

Our numerical experiments show that DFOp can solve most private black-box optimization problems exactly. Besides, it needs fewer function evaluations and achieves higher accuracy than the other compared derivative-free solvers. DFOp performs best among the five solvers when solving the private problems with noises. The performances of DFOp and NEWUOA-N are similar  to each other, which shows that DFOp solves the encrypted optimization problems well, eliminating the interference from the noise. In addition, DFOp has been applied for the private simulation experiments on the traveling wave tube design, with good results.  

\section{Conclusions and the Future Work}

This paper analyzes the least Frobenius norm updating of the quadratic model functions in derivative-free optimization with transformed/encrypted objective functions, and makes improvements for solving black-box problems with transformed/encrypted objective functions and private black-box optimization problems. Due to the fact that original model updating formula can not be applied to the situation where the output objective function changes with the iteration, a new model updating formula has been proposed, and the improved derivative-free solver based on such model updating formula is named as DFOp. Convergence analysis of DFOp for first-order critical points is given.  
In addition, the discussion about model functions and transformations/encryptions is given, with the analysis of the solution shift.
Moreover, to encrypt the function value of $h(x)$ in the private black-box optimization, two optimality-preserving differentially private noise-adding mechanisms have been designed, and the properties that they are both differentially private with a controlled privacy budget are elucidated as well. At the end of the paper, we present our numerical results which indicate that DFOp is better than other compared derivative-free solvers for solving the problems above.

Furthermore, it is obvious and natural to find that the target private problems can include more than two parts. The objective function can be extended to multiple functions, and then the private black-box optimization problem has the generated form as 
\begin{equation*}
\min_{x \in \mathbb{R}^n} \ f(x)+\sum_{i=1}^{J} l_j(x),
\end{equation*}
where $f$ is the shared and public function, each $l_j$ is a private black-box function, and $J$ is the number of the cooperations, which can provide the function evaluations of their own functions $l_j, \ j=1,\ldots,J$. DFOp can be used to solve the private black-box optimization problems defined above. Therefore, the discussion and the solver in this paper can be extended to game theory in some senses. Derivative-free solvers for constrained black-box optimization problems with transformed/encrypted objective functions can be considered as well.
 
We believe that the approach and analysis we proposed can also be expanded to improve many other derivative-free optimization solvers.  
The combination of the derivative-free optimization and the private mechanisms waits for more researches in the future.
 
%%===========================================================================================%%
% \bibliographystyle{siam} 
% \bibliographystyle{siam} 
\bibliographystyle{abbrv}      
\bibliography{sn-bibliography}% common bib file

\end{document}